\def\draft{0}
\documentclass[11pt]{article}
\usepackage[dvips,pagebackref,letterpaper=true,colorlinks=true,pdfpagemode=none, linkcolor=blue,citecolor=blue,pdfstartview=FitH]{hyperref}
\usepackage[margin=1in,nohead]{geometry}
\usepackage{dsfont}
\usepackage{paralist}
\usepackage{xspace}

\usepackage{amsmath,amsfonts,amssymb}
\newtheorem{theorem}{Theorem}[section]
\newtheorem{nt}{Selfnote}
\newtheorem{conjecture}[theorem]{Conjecture}
\newtheorem{definition}[theorem]{Definition}
\newtheorem{lemma}[theorem]{Lemma}

\newtheorem{claim}[theorem]{Claim}

\newtheorem{remk}[theorem]{Remark}
\newtheorem{exmp}[theorem]{Example}

\newenvironment{remark}{\begin{remk}
\begin{normalfont}}{\end{normalfont}
\end{remk}}

\def\FullBox{\hbox{\vrule width 8pt height 8pt depth 0pt}}

\def\qed{\ifmmode\qquad\FullBox\else{\unskip\nobreak\hfil
\penalty50\hskip1em\null\nobreak\hfil\FullBox
\parfillskip=0pt\finalhyphendemerits=0\endgraf}\fi}

\def\qedsketch{\ifmmode\Box\else{\unskip\nobreak\hfil
\penalty50\hskip1em\null\nobreak\hfil$\Box$
\parfillskip=0pt\finalhyphendemerits=0\endgraf}\fi}

\newenvironment{proof}{\begin{trivlist} \item {\bf Proof:~~}}
  {\qed\end{trivlist}}

\newcommand{\etal}{{ et~al.\ }}

\newcommand{\R}{{\mathbb R}}

\newcommand{\Z}{{\mathbb Z}}

\newcommand{\pmone}{\{-1,1\}}

\def\abs#1{\left| #1 \right|}
\newcommand{\norm}[1]{\ensuremath{\left\lvert #1 \right\rvert}}
\newcommand{\mydot}[2]{\ensuremath{\left\langle #1, #2 \right\rangle}}

\newcommand{\one}{{\mathds{1}}}
\newcommand{\ind}[1]{\one_{\{#1\}}}

\renewcommand{\Pr}{\mathop{\mathbb P}\displaylimits}
\newcommand{\pr}{\mathop{\mathbb P}\displaylimits}

\let\av=\Exp

\newcommand{\inparen}[1]{\left(#1\right)}             %\inparen{x+y}  is (x+y)
\newcommand{\inbrace}[1]{\left\{#1\right\}}           %\inbrace{x+y}  is {x+y}
\newcommand{\insquar}[1]{\left[#1\right]}             %\insquar{x+y}  is [x+y]
\newfont{\inhead}{eufm10 scaled\magstep1}

\newcommand{\suchthat}{{\;\; : \;\;}}

\newcommand{\UGC}{{\sf Unique Games Conjecture}\xspace}
\newcommand{\ULC}{{\sf Unique Label Cover}\xspace}
\newcommand{\LC}{{\sf Label Cover}\xspace}

\ifnum\draft=1
\newcommand{\mnote}[1]{[{\bf Madhur's Note: #1]}}
\else
\newcommand{\mnote}[1]{}
\fi

\ifnum\draft=1
\newcommand{\Nnote}[1]{[{\bf Nisheeth's Note: #1]}}
\else
\newcommand{\Nnote}[1]{}
\fi

\newcommand{\size}[1]{\left|#1\right|}

\newcommand{\expec}[1]{\mathbb{E} \left[#1\right]}

\newcommand{\var}[1]{{\sf Var} \left[#1\right]}
\newcommand{\prob}[1]{{\mathbb{P}} \left[#1\right]}

\def\sph{\mathbb{S}}

\def\argmin{\operatorname{argmin}}

\renewcommand{\norm}[1]{\ensuremath{\left\lVert #1 \right\rVert}}
\newcommand{\instance}{{\ensuremath{{\mathcal U} = (V,W,E)}}\xspace}
\renewcommand{\L}{{\cal L}}
\newcommand{\U}{{\cal U}}
\newcommand{\val}{\mathsf{val}}
\newcommand{\opt}{\mathsf{opt}}

\newcommand{\bfb}{{\bf b}}
\newcommand{\bfz}{{\bf z}}
\newcommand{\bfy}{{\bf y}}
\newcommand{\bfx}{{\bf x}}
\newcommand{\subspace}[2]{\textsf{Subspace}(#1,#2)}
\newcommand{\tr}{\mathsf{Tr}}
\newcommand{\fb}[1]{f_{b_{#1}}}
\newcommand{\fbpi}[1]{f_{\pi(b_{#1})}}
\newcommand{\nkfactor}{\alpha_{n,k}}

\usepackage{nicefrac}
\newcommand{\nfrac}{\nicefrac}
\newcommand{\defeq}{\stackrel{\textup{def}}{=}}
\renewcommand{\epsilon}{\varepsilon}

\renewcommand{\bfz}{z}
\renewcommand{\bfy}{y}
\renewcommand{\bfx}{x}

\begin{document}

\title{Algorithms and Hardness for Subspace Approximation\ifnum\draft=1 \\
\textsc{\small Working Draft: Please Do Not Distribute}\fi}
\author{Amit Deshpande\thanks{Microsoft Research India. {\tt amitdesh@microsoft.com}.}
\and Madhur Tulsiani\thanks{Institute for Advanced Study. {\tt madhurt@math.ias.edu}. This material
is based upon  work supported by the National Science Foundation under grant No.  CCF-0832797
and IAS Sub-contract no. 00001583. Work partly done while visiting Microsoft Research India.}
% \and Kasturi Varadarajan\thanks{Dept. of Computer Science, The University of Iowa. {\tt kvaradar@cs.uiowa.edu}}.
\and Nisheeth K. Vishnoi\thanks{Microsoft Research India. {\tt nisheeth.vishnoi@gmail.com}}}

\begin{titlepage}
\maketitle

\begin{abstract}
The subspace approximation problem \subspace{$k$}{$p$} asks for a $k$-dimensional linear subspace that fits a given set of $m$ points in $\R^n$ optimally. The error for fitting is a generalization of the least squares fit and uses the $\ell_{p}$ norm of the ($\ell_2$) distances of the points from the subspace. Most of the previous work on subspace approximation has focused on small or constant $k$ and $p = 1$ or $\infty$, using coresets and sampling techniques from computational geometry.

\medskip
In this paper, extending another line of work based on convex relaxation and rounding, we give a polynomial time algorithm, \emph{for any $k$ and any $p \geq 2$}.This extends a result of Varadarajan, Venkatesh, Ye and Zhang \cite{VVYZ}, who gave an $O(\sqrt{\log m})$ approximation for all $k$ and $p = \infty$. The approximation guarantee of our algorithm is roughly $\sqrt2 \gamma_p$ where $\gamma_p \approx \sqrt{p/e}(1+o(1))$ is the
$p^th$ norm of a standard normal random variable. The approximation ratio improves to $\gamma_p$ in the 
interesting special case when $k = n-1$. We also show that the convex relaxation we use has an integrality gap (or ``rank gap'') of $\gamma_{p} (1 - \epsilon)$, for any constant $\epsilon > 0$.

\medskip
We also study the hardness of approximating this problem. We show that assuming the Unique Games Conjecture, the subspace approximation problem is hard to approximate within a factor better than $\gamma_{p} (1 - \epsilon)$, for any constant $\epsilon > 0$. The hardness reduction involves a dictatorship test which is somewhat different
from ``long-code'' based tests used in reductions from Unique Games, and seems better suited for problems
of a continuous nature.
\end{abstract}

\thispagestyle{empty}
\vfill
\noindent \textbf{Keywords:} approximation algorithms, convex programming, unique games

\end{titlepage}
%%%%%%%%%%%%%%%% INTRODUCTION %%%%%%%%%%%%%%%%%%%%%%

\section{Introduction}
Large data sets that arise in data mining, machine learning, statistics and computational geometry problems are naturally modeled as sets of points in a high-dimensional Euclidean space. Even though these points live in a high-dimensional space, in practice they are observed to have low intrinsic dimension and it is an algorithmic challenge to capture their underlying low-dimensional structure. The subspace approximation problem described below generalizes several problems formulated in this context.

\noindent \underline{\subspace{$k$}{$p$}}: Given points $a_{1}, a_{2}, \dotsc, a_{m} \in \R^{n}$, an integer $k$, with $0 \leq k \leq n$, and $p \geq 1$, find a $k$-dimensional linear subspace that minimizes the sum of $p$-th powers of Euclidean distances of these points to the subspace, or equivalently,
\[
\underset{V \suchthat \dim(V)=k}{\text{minimize}} \quad \left(\sum_{i=1}^{m} d(a_{i}, V)^{p}\right)^{\nfrac{1}{p}}.
\]

Note that, here, $\ell_{p}$ norm is used as a function of $(d(a_{1}, V), d(a_{2}, V), \dotsc, d(a_{m}, V))$;
the individual distances $d(a_{i}, V)$ are the usual $\ell_{2}$ distances.

We describe below the special cases of the subspace approximation problem which have been studied previously
and the known results about them.
\begin{enumerate}
\item \textbf{Low-rank matrix approximation problem or the least squares fit ($p=2$)}:
Given a matrix $A \in \R^{m \times n}$ and $0 \leq k \leq n$,
the matrix approximation problem is to find another matrix $B \in \R^{m \times n}$
of rank at most $k$ that minimizes the Frobenius (also known as Hilbert-Schmidt)
norm of their difference
$\norm{A-B}_{F} \defeq \left(\sum_{ij} (A_{ij} - B_{ij})^{2}\right)^{\nfrac{1}{2}}$.
Taking the rows of $A$ to be points $a_{1}, a_{2}, \dotsc, a_{m} \in \R^{n}$,
the above problem is equivalent to the problem \subspace{$k$}{$2$}.
Elementary linear algebra
shows that the optimal subspace is spanned by the top $k$ right singular
vectors of $A$, which can be found in time $O(\min\{mn^{2}, m^{2}n\})$ using
Singular Value Decomposition (SVD) \cite{GVL}.

\item \textbf{Computing radii of point sets ($p = \infty$)}:
Given points $a_{1}, a_{2}, \dotsc, a_{m} \in \R^{n}$, their \emph{outer $(n-k)$-radius}
is defined as the minimum, over all $k$-dimensional linear subspaces, of the maximum
Euclidean distance of these points to the subspace (which is equivalent to
\subspace{$k$}{$\infty$}). Gritzmann and Klee initiated the study of this quantity in
computational convex geometry \cite{GK} and gave a polynomial time algorithm for the
minimum enclosing ball problem (or the problem \subspace{0}{$\infty$}).

\begin{enumerate}
\item \emph{For small $k$:} B\u{a}doiu, Har-Peled and Indyk \cite{BHPI} gave
a $(1+\epsilon)$-approximation algorithm running in polynomial time
for the minimum enclosing
cylinder problem (equivalent to \subspace{1}{$\infty$}), which was further
extended by Har-Peled and Varadarajan \cite{HPV} to \subspace{$k$}{$\infty$}
for constant $k$.

\item \emph{For large $k$:} Brieden, Gritzmann and Klee \cite{BGK} showed that it
is NP-hard to approximate the width of a point set (equivalent to \subspace{$n-1$}{$\infty$})
within any constant factor. From the algorithmic side, the results by Nesterov
\cite{N} and Nemirovski, Roos and Tarlaky \cite{NRT} on quadratic optimization
imply $O(\sqrt{\log m})$-approximation for \subspace{$n-1$}{$\infty$} in polynomial time.
Building on these techniques, Varadarajan, Venkatesh, Ye and Zhang \cite{VVYZ}
gave a polynomial time $O(\sqrt{\log m})$-approximation algorithm for \subspace{$k$}{$\infty$},
for any $k$.
On the hardness side, they proved that there exists a constant
$\delta > 0$ such that, for any $0 < \epsilon < 1$ and $k \leq n - n^{\epsilon}$, there is no polynomial time algorithm that gives $(\log m)^{\delta}$-approximation for \subspace{$k$}{$\infty$} unless
${\rm NP} \subseteq {\rm DTIME}\left(2^{\mathrm{polylog}(n)}\right)$.
\end{enumerate}
%\mnote{Should the $\log m$ factors in approximation above be $\log n$? Also, does $(1+\epsilon)$
%approximation mean a PTAS, an FPTAS, or a $(1+\epsilon)$-approximation algorithm for some
%fixed small $\epsilon$?}

\item \textbf{Other values of $p$}:
For general $p$ and constant $k$, a result of Shyamalkumar and Varadarajan
\cite{SV} and subsequent work by Deshpande and Varadarajan \cite{DV} gave a
$(1+\epsilon)$-approximation algorithm with running time
$O\left(mn \cdot \exp(k, p, \nfrac{1}{\epsilon})\right)$.
The running time was recently improved to
$O\left(mn \cdot \mathrm{poly}(k,  \nfrac{1}{\epsilon}) + (m+n) \cdot \exp(k,  \nfrac{1}{\epsilon})\right)$
by Feldman, Monemizadeh, Sohler and Woodruff \cite{FMSW}, for the case $p=1$.
\end{enumerate}

For $p \neq 2$, we do not know any suitable generalization of SVD, and therefore, have no exact characterization of the optimal subspace. The approximation techniques used so far to overcome this are:
\begin{inparaenum}[(i)]
\item coresets and sampling-based techniques: which give nearly optimal approximations but only for small or constant $k$ and $p$.
\item convex relaxations and rounding: which give somewhat sub-optimal approximations mostly for large values of $k$; the only exception is the result of Varadarajan, Venkatesh, Ye and Zhang \cite{VVYZ} which works for any $k$
(but only for $p=\infty$).
\end{inparaenum}

\subsubsection*{Our work}
In this paper, we study the problem \subspace{$k$}{$p$} for $p < \infty$, about which
little is known in general. One motivation for doing so is that often
the case $p<\infty$ gives significantly better approximation guarantees and requires
somewhat different techniques to analyze than $p = \infty$.
This is evident in the work for subspace approximation for small $k$ (\cite{DV} and
\cite{FMSW} for $p < \infty$ versus \cite{BHPI} and \cite{HPV} for $p = \infty$) and
in the work on regression (\cite{C05} and \cite{DDH+08} versus the $p = \infty$ case
which is solvable by fixed dimensional linear programming). Also, in the study of hardness
of approximation, the case $p = \infty$ can often be reduced to a discrete problem; while
the case $p < \infty$ is inherently of a more continuous nature, and requires somewhat
different techniques.

On the algorithmic side, we give a factor $\gamma_p \cdot \sqrt{2 - (\nfrac{1}{n-k})}$
approximation algorithm for the problem \subspace{$k$}{$p$} in $\R^n$, where
$\gamma_p \approx \sqrt{p/e}(1+o(1))$ is the $p^{th}$ norm of a standard Gaussian. 
Our algorithm is based on
a convex relaxation, similar to the semi-definite relaxations used in \cite{NRT}
and \cite{VVYZ} for $p = \infty$. We give a tighter analysis for general $p$. We
also exhibit gap instance for the convex program. We show a gap of 
factor $\gamma_p$ for \subspace{$k$}{$p$} (when $k$ is superconstant) 
showing that our analysis is tight up to the factor of
$\sqrt{2 - (\nfrac{1}{n-k})}$.

We also investigate the hardness of approximation for \subspace{$k$}{$p$}.
We give a reduction from the \ULC problem of Khot \cite{Khot02} to
the problem of approximating \subspace{$n-1$}{$p$} within a factor $\gamma_p$
(which can trivially be extended to a reduction to \subspace{$k$}{$p$}
for $k = n^{\Omega(1)}$). The reduction is related to the ones used for similar
geometric problems in \cite{KNS}, \cite{KN08} and \cite{ABHKS05}. However,
an interesting difference here in comparison to usual reductions is that we use 
a different (real-valued) encoding of
the assignment to \ULC (in terms of the Fourier coefficients of the long-code
instead of the truth table) which is more natural in our context. This may also 
be useful for other problems of a continuous nature.

Very recently, our techniques were also extended by Guruswami \etal \cite{GRSW10} 
to give a  reduction from the \LC problem (\emph{without} assuming the uniqueness property) 
to approximating \subspace{$n-1$}{$p$} within a factor of $\gamma_p$. 
This proves an unconditional NP-hardness for the latter problem.

\subsubsection*{Other related problems}
\paragraph{$L_p$-Grothendieck problem.}
In the $k=n-1$ case, subspace approximation problem can be rewritten as
$\min_{\norm{z}_{2}=1} \norm{Az}_{p}$,
where the rows of $A \in \R^{m \times n}$ represent the points $a_{1}, a_{2}, \dotsc, a_{m}$ and $z \in \R^{n}$ represents the unit normal to the subspace we are asked to find. When $A$ is invertible, this problem can be shown (using duality in Banach spaces) to be equivalent to a special case of the $L_{p}$-Grothendieck problem (introduced by Kindler, Naor and Schechtman \cite{KNS}) which asks for maximizing $x^TMx$ subject to
$\norm{x}_p \leq 1$. $\subspace{n-1}{p}$ with invertible $A$, reduces to this problem with 
$M = (A^{-1})^TA^{-1}$.

In this special case, using Grothendieck's inequality and a technique by Alon and Naor \cite{AN}, one can get $O(1)$-approximation. Moreover, in this case, the above problem is also equivalent to finding diameters of convex bodies given by $\norm{Ax}_{p} \leq 1$ and computing $p \mapsto 2$ norm of the matrix $A^{-1}$.
%%
%In the standard Grothendieck problem \cite{PisierFactorization}, given a matrix $A,$ one is  required to compute  $\|A\|_{\infty \mapsto 1}.$ (In general $\|A\|_{p \mapsto q} \defeq \sup \{ \|Ax\|_q \suchthat \|x\|_p \leq 1 \}.$)
%%
%In general, computing this norm is NP-hard.
%Grothendieck, in his ``little'' paper proved (implicitly), that one can write down a semi-definite programming (SDP) relaxation for $\|A\|_{\infty \mapsto 1},$ such that the value  of this relaxation is no more than an absolute constant times $\|A\|_{\infty \mapsto 1}.$ Alon and Naor consider the SDP relaxation explicitly and show how to round its solution in  polynomial time to obtain  a constant factor approximation for computing $\|A\|_{\infty \mapsto 1}.$ It  is easy to see that for $p \geq 2$ and $q \leq 2,$ one can again write a similar SDP relaxation for $\|A\|_{p \mapsto q}$ and round it to obtain a constant-factor approximation for this quantity. Now, one notices that when $q=2,$ the SDP relaxation of Alon and Naor is the same as the $L_p$-Grothendieck relaxation of Kindler {\em et al.} \cite{KNS}.
%%
%Thus, in this special case, using Grothendieck's Inequality and a technique by Alon and Naor \cite{AN}, one can get $O(1)$-approximation for $p\geq 2.$  Moreover, in this case, the above problem is also equivalent to finding the diameter of a convex body given by $\{x:\norm{Ax}_{p} \leq 1\}$ and computing $p \mapsto 2$ norm of the matrix $A^{-1}$.
%\mnote{Need more details for the above section.}

\vspace{-2mm}
\paragraph{$l_p$-regression problem.}
In the $l_p$ regression problem, we are given an $m \times n$ matrix $A$ and
a vector $b \in \R^m$, and the goal is to minimize $||Az-b||_p$ over all
$z \in \R^n$. This is clearly related to subspace approximation with
$k = n-1$, but the fact that $z$ is unconstrained makes it a convex
optimization problem. Efficient approximation algorithms for the regression
problem are given by Clarkson \cite{C05} for $p = 1$, Drineas, Mahoney, and
Muthukrishnan \cite{DMM06} for $p = 2$, and Dasgupta et al. \cite{DDH+08}
for $p \geq 1$. It is not clear that these results can be employed
fruitfully for the subspace approximation problem for $k = n - 1$ where it is
required that $\norm{z} \geq 1$.

%%%%%%%%%%%%%%%%%%%%%%%%%%%%%%%%%%%% PRELIMINARIES
\section{Preliminaries and Notation} \label{sec:prelims}
Throughout this paper, $\norm{\cdot}_p$ denotes the $p$-norm. Norms of vectors are taken
with respect to the counting measure and of functions are taken with respect to the
uniform probability measure on their domain. When the subscript is unspecified, 
$\norm{\cdot}$ denotes $\norm{\cdot}_2$.

\subsection{The Subspace Approximation Problem}
We will use a formulation of the problem \subspace{$k$}{$p$}
for points $a_1, \ldots, a_m$, in terms of the orthogonal complement of the desired
subspace $V$. Let $\bfz_1, \ldots, \bfz_{n-k}$ be an orthonormal basis for the orthogonal complement
and let $Z \in \R^{n \times (n-k)}$ denote the matrix with the $j^{th}$ column
$Z^{(j)} = \bfz_j$. Then $d(a_i,V) = \norm{a_i^T Z}_2$ and the problem of
finding (the orthogonal complement of) the subspace can be stated as
\begin{align*}
\text{minimize}\quad\quad &\left(\sum_{i = 1}^m \norm{a_i^T Z}_2^p\right)^{\nfrac{1}{p}} \\
\text{subject to:}\quad\quad &\norm{Z^{(j)}} \geq 1 ~\forall j \in \{1, \ldots, n-k\}\\
& \mydot{Z^{(j_1)}}{Z^{(j_2)}} = 0 ~\forall j_1 \neq j_2, ~~~Z \in \R^{n \times (n-k)}
\end{align*}

For the hardness results we shall be concerned with the special case of the problem
with $k = n-1$. For points $a_1, \ldots, a_m \in \R^n$, let $A$ be $m \times n$
matrix with $A_i = a_i^T$. The problem \subspace{$n-1$}{$p$} is then simply to
minimize $\norm{A z}_p$ for $z \in \R^n$, subject to $\norm{z}_2  \geq 1$.

\begin{remark}
It is easy to check
that (by a change of variable and suitable modification of $A$) both the norms can be
taken to be with respect to an arbitrary measure instead of the counting measure.
In particular, if $A \in \R^{m\times n}$, the $p$-norm is taken with respect
to a measure $\mu$ on $[m]$ and the 2-norm with respect
a measure $\nu$ on $[n]$, then we change variables to $\tilde{\bfz}$ with
$\tilde{z}_j = \sqrt{\nu(j)}z_j$ and modify $A_{ij}$ to
$A_{ij} (\mu(i))^{1/p} / \sqrt{\nu(j)}$ to get an equivalent problem
with norms according to the counting measure.
 \end{remark}

\subsection{Bernoulli and Gaussian Random Variables}
A Bernoulli random variable is a discrete random variable taking values in $\{-1,1\}$ with probability $\nfrac{1}{2}$ each. A standard normal random variables (or $1$-dimensional Gaussian) is a continuous random variable with probability density function $1/\sqrt{2\pi} \cdot \exp(-x^{2}/2)$. We use $\gamma_p$ to denote the $p^{th}$ moment of $N(0,1)$,
\[
\gamma_{p} \defeq  \left(\int_{-\infty}^{\infty} \abs{x}^{p} \cdot \frac{e^{-\nfrac{x^{2}}{2}}}{\sqrt{2\pi}}  dx\right)^{\nfrac{1}{p}}
=~ \left(\frac{2^{\nfrac{p}{2}} \cdot \Gamma\left(\frac{p+1}{2}\right)}{\sqrt{\pi}}\right)^{\nfrac{1}{p}}
\approx~~ \sqrt{\frac{p}{e}}(1+o(1)).
\]

We shall require both upper and lower bounds on moments of a sum of Bernoulli
random variables by the moment of an appropriate Gaussian. The following
upper bound is one direction of the Khintchine inequality
(see \cite{PS95}) well-known in functional analysis.

\begin{claim}\label{clm:upper-gaussian}
Let $x_1, \ldots, x_R$ be independent Bernoulli random variables and let
$c_1, \ldots, c_R \in \R$ and $\norm{\bf c} = \sqrt{c_1^2 + \cdots + c_R^2}$.
Then for any positive $p > 0$,
\[\av_{x_1, \ldots, x_R } \insquar{\inparen{\sum_{i=1}^R c_i \cdot x_i}^p}
  ~~\leq~~ \gamma_p^p \cdot \norm{\bf c}^{p} \]
\end{claim}
% \begin{proof}
% The claim follows simply by considering a multinomial expansion. We have that
% \begin{equation*}
% \av_{x_1, \ldots, x_R} \insquar{\inparen{\sum_{i=1}^R c_i \cdot x_i}^p}
% ~=~ \sum_{p_1, \ldots, p_R} \binom{p}{p_1, \ldots, p_R} \prod_R \av_{x_i}\insquar{(c_i \cdot x_i)^{p_i}}.
% \end{equation*}
% We now replace the Bernoulli variables by independent standard Gaussians $g_1, \ldots, g_R$.
% This is possible because for a term in which any $p_i$ is odd,
% $\prod_R \av_{x_i}\insquar{(c_i \cdot x_i)^{p_i}} = \prod_R \av_{g_i}\insquar{(c_i \cdot g_i)^{p_i}} = 0$
% and when all $p_i$'s are even, the coefficient $\prod_R c_i^{p_i}$ is positive and for
% each $i$, $\av\insquar{x_i^{p_i}} = 1 \leq \av\insquar{g_i^{p_i}}$. This gives
% \begin{equation*}
% \av_{x_1, \ldots, x_R} \insquar{\inparen{\sum_{i=1}^R c_i \cdot x_i}^p}
% ~\leq~ \av_{g_1, \ldots, g_R} \insquar{\inparen{\sum_{i=1}^R c_i \cdot g_i}^p}.
% \end{equation*}
% Since $\sum_{i}{c_i \cdot g_i}$ is also a Gaussian random variable with mean 0 and variance
% $\norm{\bf c}^2$, this proves the claim.
% \end{proof}

The following version of the reverse direction, when all $c_i$'s
are much smaller than $\norm{\bf c}$, can be derived using the Berry-Esseen
Theorem (as in \cite{vanB72}). A proof of the statement below appears in
\cite{KNS} (as Lemma 2.5).
\begin{claim}\label{clm:clt}
Let $x_1, \ldots ,x_R$ be independent Bernoulli random variables and let
$c_1, \ldots, c_R \in \R$ be such that for all $i \in [R]$, $|c_i| \leq \tau \cdot  \norm{\bf c}$
for $\tau \in (0,e^{-4})$. Then, for any $p \geq 1$,
\[\av_{x_1, \ldots, x_R } \insquar{\abs{\sum_{i=1}^R c_i \cdot x_i}^p}
~~\geq~~ \gamma_p^p \cdot \norm{\bf c}^{p} \cdot \inparen{1 - 10\tau\cdot(\log(\nfrac{1}{\tau}))^{\nfrac{p}{2}}}.
\]
\end{claim}

%%%%%%%%%%%%%%%% OVERVIEW %%%%%%%%%%%%%%%%%%%%%%%%%%%%%%%%

\section{Technical Overview}
In this section we describe our results and give a general outline of the sections that follow. 

\subsection{Algorithm for \subspace{$k$}{$p$}}
We formulate problem \subspace{$k$}{$p$} for points $a_{1}, a_{2}, \dotsc, a_{m} \in \R^{n}$ in terms of the orthogonal complement of the desired subspace $V$. Let $Z \in \R^{n \times (n-k)}$ be the matrix whose columns form an orthonormal basis for the orthogonal complement of $V$, then the distance of a point $a_{i}$ from $V$ can be written as $d(a_{i}, V) = \norm{a_i^T Z}_{2} = (a_i^T ZZ^T a_i)^{\nfrac{1}{2}}$. Note that $ZZ^{T}$ is a positive semidefinite
(p.s.d.) matrix of rank $n-k$, all of whose nonzero singular values are $1$ and whose singular vectors (the columns of $Z$) specify the (complement of the) subspace $V$.

A convex relaxation of \subspace{$k$}{$p$} is then obtained by optimizing over arbitrary positive semidefinite matrices $X$ and replacing the requirement that the matrix have rank $n-k$ by a condition on the trace of $X$
(see Figure \ref{fig:relaxation}). This is similar to the relaxations used in 
\cite{NRT,VVYZ}. The problem then reduces to giving a ``rounding algorithm" which reduces the 
rank of the matrix $X$ (which might be as large as $n$) to $n-k$, and achieves a good approximation of the objective value of the convex program.

In keeping with the intuition that the singular vectors of $ZZ^T$ span the orthogonal complement of $V$, our algorithm looks at the singular vectors of the matrix $X$ obtained by solving the convex relaxation. It then divides the singular vectors into $n-k$ ``bins", and constructs one vector for each bin by taking a 
random linear combination of vectors within each bin. 

Our algorithm described in Section \ref{sec:algo} achieves an approximation ratio of $\gamma_q \cdot (2 - \frac{1}{n-k})^{\nfrac{1}{2}}$ for \subspace{$k$}{$p$}, where $q = 2\cdot\lceil p/2 \rceil$. (See Theorem \ref{thm:algo}.)

We remark that the problem of obtaining low-rank solutions to a semidefinite program was also
considered by \cite{SYZ08}, and was addressed by simply taking random (chosen according to a Gaussian) 
linear combinations of the singular vectors of the relevant matrix. However, in their case, they were'
only interested in satisfying the constraints, with an error depending inversely on the rank parameter.
In our case, we require  a rank $n-k$ positive semidefinite matrix, all of whose eigenvalues are exactly 1.
Since the only constraint enforcing this is a constraint on the trace of the matrix, even a small
multiplicative error in satisfying the constraint can make some singular values quite small.
To resolve this, we proceed by dividing the singular vectors in various
bins and take Bernoulli linear combinations, do directly generate the orthogonal singular vectors. 

\subsection{A gap instance}
In Section \ref{sec:gap}, we show that the convex relaxation we use has an integrality gap, or more correctly ``rank gap'', of $\gamma_{p} (1-\epsilon)$, for any constant $\epsilon > 0$. Given any constant $\epsilon > 0$, we construct points $b_{1}, b_{2}, \dotsc, b_{m} \in \R^{n}$ such that the optimum for \subspace{$n-1$}{$p$} on these points
(a rank 1 p.s.d. matrix)
and the optimum for its corresponding convex relaxation (a rank $n$ p.s.d. matrix) are at least a
factor of $\gamma_{p} (1-\epsilon)$ apart. We first show such a gap for the continuous analog of
\subspace{$n-1$}{$p$} where the point set is the entire $\R^{n}$ equipped with Gaussian measure (Theorem \ref{thm:cont-gap}). We then discretize this example to get our final integrality gap construction (Theorem \ref{thm:discrete-gap}).

This also gives a gap of factor $\gamma_p (1-\epsilon)$ for \subspace{$k$}{$p$} for any super-constant $k = k(n)$, since an instance of \subspace{$n-1$}{$p$} in $\R^n$ can be trivially converted (by adding extra zero coordinates) to an instance of \subspace{$k$}{$p$} in $\R^{n'}$ with $k(n') = n-1$.

\subsection{Unique-Games hardness}
In Section \ref{sec:ug-hard}, we describe a reduction from \ULC to the problem of
approximating \subspace{$n-1$}{$p$} within a factor better than $\gamma_p$ (for a constant $p \geq 1$).
By a trivial reduction from \subspace{$n-1$}{$p$} to \subspace{$k$}{$p$} for any $k = n^{\Omega(1)}$,
this gives the hardness of approximating \subspace{$k$}{$p$} better than $\gamma_p$, assuming
the \UGC.

To understand the intuition for the reduction, let us consider the simpler
problem of testing whether a given function $f:\pmone^R \to \pmone$ is a
``dictator" i.e. $f(x_1, \ldots, x_R) = x_i$ for some $i \in [R]$, which
is a useful primitive in such reductions. The problem is to design an instance
$\cal I$ of \subspace{$n-1$}{$p$} and interpret the description of $f$ as a solution
to $\cal I$. The required property is that if $f$ is a
dictator then the corresponding subspace fits the points in $\cal I$ with small error.
On the other hand, if $f$ is ``far from being a dictator", the error is required to be larger by a factor of
$\gamma_p$.

In most reductions, $f$ is assumed to be described by its truth table. However, if we
want to interpret the input simply as the coordinates of a vector $z$, there is no way
to enforce that the coordinates be boolean. It turns out to be more convenient if we
require $f$ as a list of its Fourier coefficients which can be thought of as a vector with arbitrary
real numbers coordinates and norm 1 (since $\av[f^2] = 1$). Also, since we are only
interested in dictator functions, it is sufficient to ask for the ``level 1" Fourier coefficients
$\hat{f}(\{1\}), \ldots, \hat{f}(\{R\})$.

In particular, consider the input being described by $R$ real numbers $b_1, \ldots, b_R$
such that $\sum_i b_i^2 = 1$ and we think of it as describing the function
$f_b(x_1, \ldots, x_R) = b_1 \cdot x_1 + \cdots + b_R \cdot x_R$. We also think of
$b_1, \ldots, b_R$ as the normal to some $R-1$ dimensional subspace of
$\R^R$. Let the points be given by $(\nfrac{1}{2^{\nfrac{R}{p}}}) \cdot x$ for each
vector $\bfx \in \pmone^R$,
so that the objective of the subspace approximation problem is exactly $\norm{f_b}_p$.
If $f$ is a dictator, i.e., one of the $b_i$'s is 1 and others 0, then $\norm{f_b}_p  = 1$.
Also, if it is far from a dictator in the sense that $\max_ib_i \leq \tau$ for a
small constant $\tau$, then $\norm{f_b}_p \approx \gamma_p$ by Claim \ref{clm:clt}.

Translating this intuition to a reduction from \ULC turns out to be slightly technical
due to the fact that we need to consider one function for each vertex of \ULC and
all bounds on norms do not hold for individual functions but only on average. Similar technicalities
arise when working with the $\ell_p$  norm in \cite{KNS}
(though they specify functions by their truth tables).

%%%%%%%%%%%%%%%% ALGORITHM %%%%%%%%%%%%%%%%%%%%%%
\section{Approximation Algorithm via Convex Programming} \label{sec:algo}
To relax the minimization problem for \subspace{$k$}{$p$}  to a convex problem, we rewrite the distances $\norm{a_i^T Z}$ in the objective as $(a_i^T Z Z^T a_i)^{\nfrac{1}{2}}$. Noting that $ZZ^T$ is a positive semidefinite matrix of rank $n-k$, we get the following natural relaxation similar to the one used in \cite{NRT,VVYZ}.
\begin{figure}[h]
\hrule
\vline
\begin{minipage}[t]{0.49\linewidth}
\vspace{10 pt}
\centering {\underline{\textsf{Minimization Problem}}}
\begin{align*}
\text{minimize}\quad\quad &\left(\sum_{i = 1}^m \abs{a_i^T Z Z^T a_i}^{\nfrac{p}{2}}\right)^{\nfrac{1}{p}} \\
\text{subject to:}\quad\quad &\norm{Z^{(j)}} \geq 1 ~\forall j \in \{1, \ldots, n-k\}\\
  & \mydot{Z^{(j_1)}}{Z^{(j_2)}} = 0 ~\forall j_1 \neq j_2\\
  & Z \in \R^{n \times (n-k)}
\end{align*}
\end{minipage}
\vline
\begin{minipage}[t]{0.49\linewidth}
\vspace{10 pt}
\centering {\underline{\textsf{Convex Relaxation}}}
\begin{align*}
\text{minimize}\quad\quad &\left(\sum_{i = 1}^m \abs{a_i^T X a_i}^{\nfrac{p}{2}}\right)^{\nfrac{1}{p}} \\
\text{subject to:}\quad\quad &\tr(X) \geq n-k\\
  & I \succcurlyeq X \succcurlyeq 0 \\
  & X \in \R^{n \times n}
\end{align*}
\end{minipage}
\hfill\vline
\hrule
\caption{The problem $\subspace{k}{p}$ and its convex relaxation}
\label{fig:relaxation}
\end{figure}

Note that this relaxation removes the constraint on the rank and relaxes the constraint on the length of the
individual vectors $Z^{(j)}$ to the trace of entire matrix $X$. Also, the objective function is written as
$\inparen{\sum_i \abs{a_i^T X a_i}^{\nfrac{p}{2}}}^{\nfrac{1}{p}}$ which is not convex.
However, for solving the convex program, we can work with $\sum_i \abs{a_i^T X a_i}^{\nfrac{p}{2}}$, which is convex for $p \geq 2$.

In Figure \ref{fig:algo}, we give a ``rounding algorithm'' for the relaxation. Note that the problem
here is not really to round the solution to an integer solution as with most
convex relaxations, but instead to {\em reduce the rank} of the solution to the
program, while obtaining a good approximation of the objective.

\begin{figure}[htb]
\hrule
\vline
\begin{minipage}[t]{0.98\linewidth}
\vspace{10 pt}
\begin{center}
\begin{minipage}[h]{0.95\linewidth}
% {\small
\underline{\textsf{Input}}: A matrix $X \in \R^{n \times n}$ satisfying
$I \succeq X \succeq 0$ and $\tr(X) \geq n-k$.

\begin{enumerate}
\item Express $X$ in terms of its singular vectors as $X = \sum_{t=1}^r \lambda_t \bfx_t \bfx_t^T$
where the vectors $\bfx_1, \ldots, \bfx_r$ form an orthonormal set
and $\lambda_1 \geq \lambda_2 \geq \cdots \geq \lambda_r \geq 0$.

\item Partition $[r]$ into $n-k$ subsets $S_1, \ldots, S_{n-k}$. Start with $S_1 = \cdots = S_{n-k} = \emptyset$.
Then for $t$ from 1 to $r$  do:
\begin{enumerate}
\item Find the set $S_{j}$ for which $\sum_{t' \in S_{j}} \lambda_{t'}$ is minimum.
\item Set $S_j := S_j \cup \{t\}$.
\end{enumerate}

\item Pick $r$ independent Bernoulli variables $b_1, \ldots, b_r \in_R \{-1,1\}$. For each $j \in [n-k]$,
let $\bfy_j \defeq \sum_{t \in S_j} b_t \cdot \sqrt{\lambda_t} \cdot \bfx_t $.

\item Output the matrix $Z \in \R^{n \times (n-k)}$ with $Z^{(j)} \defeq \dfrac{\bfy_j}{\norm{\bfy_j}}$.
\end{enumerate}

\vspace{5 pt}
% }
\end{minipage}
\end{center}
\end{minipage}
\hfill \vline
\hrule
\caption{The rank reduction algorithm}
\label{fig:algo}
\end{figure}

We shall show the algorithm outputs a matrix $Z$ of rank $n-k$
which achieves an approximation ratio of $\gamma_p \cdot \sqrt{2 - \nfrac{1}{n-k}}$
in expectation, for even integers $p \geq 2$.
An approximation guarantee for other values of $p$ can be obtained via Jensen's inequality.
We state the dependence on $n-k$ precisely as we shall be interested in the case $n-k = 1$.
For notational convenience, we shall use $\nkfactor$ to denote the quantity
$\sqrt{2 - \nfrac{1}{n-k}}$ in the rest of this section.

It is clear that the columns of the matrix $Z$ given by the algorithm form an orthonormal set 
since they are all in the span of distinct eigenvectors of $X$, and are normalized to have length 1. However,
this assumes that the lengths of the vectors $\bfy_j$ are nonzero. Since a vector
$\bfy_j$ is a weighted sum of orthogonal vectors, $\norm{\bfy_j}^2 = \sum_{t \in S_j} \lambda_t$.
The following claim gives a lower bound on this quantity
which is also useful in bounding the approximation ratio.

\begin{claim} \label{clm:greedy}
Let $S_1, \ldots, S_{n-k}$ be the partition constructed by the algorithm in step 2. Then
\[ \forall j \in [n-k], ~~~\sum_{t \in S_j} \lambda_t ~~\geq~~ \frac{1}{\nkfactor^2}.\]
\end{claim}
\begin{proof} Let ${j_0} \defeq \argmin_j \inbrace{\sum_{t \in S_j} \lambda_t}$ and let
$s^{*} \defeq \sum_{t \in S_{j_0}} \lambda_t$. Let $Q \defeq \{j_0\} \cup \{j ~|~ |S_j| > 1\}$.
Note that the algorithm ensures that $|S_j| > 0$ for all $j$ but in $T$ we discard the
singleton sets. We will show that $s^{*} \geq 1/\inparen{2 - \nfrac{1}{|Q|}}$, which will
prove the claim since $|Q| \leq n-k$.

We argue that for each $j \in Q$, $j \neq j_0$, $\sum_{t \in S_j} \lambda_t ~\leq~ 2s^{*}$.
To see this, let $t_j$ be the maximal index in $S_j$. At step $t = t_j$, $t_j$ was added to
set $S_j$ and not to the set $S_{j_0}$. Hence,
\[\sum_{t \in S_j, t< t_j} \lambda_t ~\leq~ \sum_{t \in S_{j_0}, t < t_j} \lambda_t ~\leq~ s^{*}.\]
Also, there exists at least one $t_0 \in S_{j_0}$ such that $t_0 < t_j$.
This is because $S_j$ was non-empty at step $t_j$ (otherwise it would be a singleton).
But then $\lambda_{t_j} \leq \lambda_{t_0} \leq s^{*}$ and, hence, $\sum_{t \in S_j} \lambda_t \leq 2s^*$.

Finally, we note that for each $j \notin Q$, $S_j$ contains exactly one element $t$, the eigenvalue
$\lambda_t$ corresponding to which is at most 1. Thus,
\[ (|Q|-1)\cdot 2s^* + s^* + (n-k-|Q|) \cdot 1 ~~\geq~~ \sum_{t \in [r]} \lambda_t ~~\geq~~ n-k,\]
which completes the proof.
\end{proof}

The following lemma proves the required approximation guarantee for the expected $p^{th}$ moment
of the distance a single  point $a_i$ from the orthogonal complement of the column span of $Z$.

\begin{lemma} \label{lemma:peven}
Let $X$ be the solution of the convex relaxation and let $Z$ be the matrix
returned by the algorithm. Also, let $p$ be even. Then, for each $i \in [m]$
\begin{equation*}
 \av_{Z} \left[ \norm{a_i^T Z}_2^p \right] ~~\leq~~
\gamma_p^p \cdot \nkfactor^p \cdot \inparen{a_i^TXa_i}^{\nfrac{p}{2}}.
\end{equation*}
\end{lemma}
\begin{proof}
We can expand $\norm{a_i^T Z}$, using $W_j$ to denote $\mydot{a_i}{Z^{(j)}}$, as
\begin{equation*}
\av_Z \left[\norm{a_i^T Z}_2^p\right]
~=~  \av_Z \left[ \left(\sum_{j=1}^{n-k} \mydot{a_i}{Z^{(j)}}^2\right)^{\nfrac{p}{2}}\right]
~=~  \av\left[ \left(\sum_{j=1}^{n-k} W_j^2\right)^{\nfrac{p}{2}} \right]
\end{equation*}
Note that the $W_j$-s are independent random variables since each $W_j$
only depends on $b_t$ such that $t \in S_j$, and the sets are disjoint.
Using the multinomial expansion and the fact that $p$ is even, the above can be written
as
\begin{eqnarray*}
 \av\left[ \left(\sum_{j=1}^{n-k} W_j^2\right)^{\nfrac{p}{2}} \right]
&=& \sum_{p_1, \ldots, p_{n-k}} \binom{p/2}{p_1, \ldots, p_{n-k}} \av\left[\prod_j W_j^{2p_j}\right]\\
&=& \sum_{p_1, \ldots, p_{n-k}} \binom{p/2}{p_1, \ldots, p_{n-k}} \left(\prod_j \av [W_j^{2p_j}]\right).
\end{eqnarray*}

The following claim then finishes the proof.
\begin{claim}
 $\av\left[W_j^{2p_j}\right] ~~\leq~~
\gamma_{p}^{2p_j}  \cdot \inparen{\dfrac{\sum_{t \in S_j} \lambda_t \mydot{a_i}{\bfx_t}^2}{\sum_{t \in S_j} \lambda_t}}^{p_j}.$
\end{claim}
\begin{proof}
The proof follows an application of upper bound on a sum Bernoulli variables
derived in Claim \ref{clm:upper-gaussian}. We expand $\av[W_j^{2p_j}]$ as
\begin{equation*}
 \av\left[ W_j^{2p_j}\right]
~=~ \av \insquar{\inparen{\frac{\mydot{a_i}{\sum_{t \in S_j} b_t \cdot  \sqrt{\lambda_t}\cdot  \bfx_t}}
{\norm{\sum_{t \in S_j} b_t \cdot \sqrt{\lambda_t} \cdot \bfx_t}}}^{2p_j}}
~=~  \frac{\av\insquar{ \inparen{\sum_{t \in S_j}b_t \cdot \sqrt{\lambda_t}\cdot \mydot{a_i}{\bfx_t}}^{2p_j} }}{\inparen{\sum_{t \in S_j} \lambda_t}^{p_j}}.
\end{equation*}
Claim \ref{clm:upper-gaussian} gives that
$\av\insquar{ \inparen{\sum_{t \in S_j}b_t \cdot \sqrt{\lambda_t} \cdot\mydot{a_i}{\bfx_t}}^{2p_j} }
~~\leq~~ \gamma_{2p_j}^{2p_j} \cdot \inparen{\sum_{t \in S_j} \lambda_t \mydot{a_i}{\bfx_t}^2}^{p_j}$
and noting that $\gamma_{2p_j} \leq \gamma_p$ (since $2p_j \leq p$) proves the claim.
\end{proof}

For each $j$, let $D_j$ denote $\sum_{t \in S_j} \lambda_t \mydot{a_i}{\bfx_t}^2$
and let $\Lambda_j$ denote $\sum_{t \in S_j} \lambda_t$. Using the above claim we get that
\begin{equation*}
\av_Z \left[\norm{a_i^T Z}_2^p\right]
 ~\leq~ \sum_{p_1, \ldots, p_{n-k}} \binom{p/2}{p_1, \ldots, p_{n-k}} \cdot
 \prod_{j} \inparen{\frac{D_j}{\Lambda_j}}^{p_j} \cdot \gamma_p^{p}
 ~=~ \inparen{\sum_j \frac{D_j}{\Lambda_j}}^{\nfrac{p}{2}} \cdot \gamma_p^p.
 \end{equation*}

Claim \ref{clm:greedy} gives that $\nfrac{1}{\Lambda_j} \leq \nkfactor^2$. Also, we have that
$\sum_j D_j = \sum_t \lambda_t \mydot{a_i}{\bfx_t}^2 = a_i^T X a_i$.
Combining these gives
$\av_Z \left[\norm{a_i^T Z}_2^p\right] ~\leq~ \gamma_p^p \cdot \nkfactor^{p} \cdot \inparen{a_i^T X a_i}^{\nfrac{p}{2}}$
which proves the lemma.
\end{proof}

An approximation guarantee for other values of $p$ can be obtained via a standard application of Jensen's Inequality. We state the dependence on $n-k$ precisely as we shall be interested in the case $n-k = 1$ in the later sections. Notice that the approximation factor is $\gamma_{q}$, where $q = 2 \cdot \lceil p/2 \rceil$, in the case $n-k=1$, and thus matches the integrality gap and unique-games hardness that appear in the later sections.
\begin{theorem} \label{thm:algo}
Let $X$ be the solution of the convex relaxation and let $Z$ be the matrix
returned by the algorithm. Let $p \geq 1$ and let $q = 2 \cdot \lceil  \nfrac{p}{2} \rceil$
be the smallest even integer such that $q \geq p$. Then,
\[\av_{Z} \left[\left(\sum_{i=1}^m \norm{a_i^T Z}_2^p \right)^{\nfrac{1}{p}}\right] ~~\leq~~
\gamma_q \cdot \sqrt{2 - (\nfrac{1}{n-k})} \cdot \left(\sum_{i=1}^m \inparen{a_i^T X a_i}^{\nfrac{p}{2}}\right)^{\nfrac{1}{p}}.\]
\end{theorem}
\begin{proof} (Proof of Theorem \ref{thm:algo})
By the concavity of the function $f(u) = u^{\nfrac{1}{p}}$ and Jensen's Inequality
we have that
\begin{equation*}
 \av_{Z} \left[\left(\sum_{i=1}^m \norm{a_i^T Z}_2^p \right)^{\nfrac{1}{p}}\right] ~~\leq~~
\left(\av_{Z} \left[\sum_{i=1}^m \norm{a_i^T Z}_2^p\right]\right)^{\nfrac{1}{p}},
\end{equation*}
and by linearity it suffices to consider a single term of the summation. Another
application of Jensen's (using $p \leq q$) and Lemma \ref{lemma:peven} give that
\begin{equation*}
 \av_{Z} \insquar{\norm{a_i^T Z}_2^p} ~=~ \av_Z \insquar{\inparen{\norm{a_i^T Z}_2^q}^{\nfrac{p}{q}}}
~\leq~ \inparen{\av_Z \insquar{\norm{a_i^T Z}_2^q}}^{\nfrac{p}{q}}
~\leq~ \gamma_q^p \cdot \nkfactor^p \cdot  \inparen{a_i^T X a_i}^{\nfrac{p}{2}}
\end{equation*}
which completes the proof of the theorem.
\end{proof}

\begin{remark}
Our results are stated in terms of the expected approximation ratio achieved by the algorithm.
However, one can get arbitrarily close to this ratio with high probability, simply by considering
few independent runs of the algorithm and picking the best solution. In particular, one can
achieve an approximation guarantee $(1+\epsilon)\cdot\gamma_q \cdot \sqrt{2 - (\nfrac{1}{n-k})}$
with probability $1-p_e$, by using $O(\nfrac{1}{\epsilon} \cdot \log(\nfrac{1}{p_e}))$ runs.
\end{remark}

%\mnote{Is there a better way to handle the case of general $p$ than using Jensen's Inequality?}
%\Nnote{I think this proof appears in the SODA 09 paper and  is pretty standard and can be put in the appendix.}

%%%%%%%%%%%%%%%% INTEGRALITY GAP %%%%%%%%%%%%%%%%%%%%%%
\section{A Gap  Instance for the Convex Relaxation}\label{sec:gap}

Here we describe an instance of \subspace{$n-1$}{$p$} such that the value of any valid solution (which is of rank $1$) is at least $\gamma_p$ times the value of the convex relaxation. Note that approximation ratio of the algorithm for the case $n-k=1$ (and even $p$) is exactly $\gamma_p$ and hence this shows that our analysis is optimal for this case.

This also gives a gap of factor $\gamma_p$ for \subspace{$k$}{$p$}
for any super-constant $k = k(n)$, since an instance of \subspace{$n-1$}{$p$}
in $\R^n$ can be trivially converted (by adding extra zero coordinates) to
an instance of \subspace{$k$}{$p$} in $\R^{n'}$ with $k(n') = n-1$.

\subsection{A continuous gap instance}
Recall that an instance of \subspace{$n-1$}{$p$} can be expressed as $\min_{\norm{z}_{2} = 1} \norm{Az}_{p}$ for $A \in \R^{n \times m}$, where $a_{1}, a_{2}, \dotsc, a_{m}$ form the rows of $A$. We consider a continuous generalization of this, where instead of points, we are given a probability distribution on $\R^n$ with density function $\mu(\cdot)$,
and objective is:
\[
\min_{\norm{z}_{2} = 1} \left(\int_{a \in \R^n} \abs{\mydot{a}{z}}^{p} \mu(a) da\right)^{\nfrac{1}{p}}.
\]
The corresponding convex relaxation is
\[
\min_{\substack{I \succcurlyeq X \succcurlyeq 0 \\ \tr(X)=1}} \left(\int_{a \in \R^n} \left(a^{T}Xa\right)^{\nfrac{p}{2}} \mu(a) da\right)^{\nfrac{1}{p}}.
\]
We first show that Gaussian measure on $\R^n$, i.e., i.i.d. coordinates from $N(0,1)$, gives a gap instance for the above problem.
\begin{theorem} \label{thm:cont-gap}
Given $\eta > 0$, there exists $n_0 \in \Z$ such that for all $n \geq n_0$
if $\mu$ is the Gaussian  density function on $\R^n$ with each coordinate
having mean 0 and variance 1, then
\begin{equation*}
\min_{\norm{z}_{2} = 1} \left(\int_{a \in \R^n} \abs{\mydot{a}{z}}^{p} \mu(a) da\right)^{\nfrac{1}{p}}
~\geq~ \gamma_p(1-\eta) \cdot
\min_{\substack{I \succcurlyeq X \succcurlyeq 0 \\ \tr(X)=1}} \left(\int_{a \in \R^n} \left(a^{T}Xa\right)^{\nfrac{p}{2}} \mu(a) da\right)^{\nfrac{1}{p}}.
\end{equation*}
\end{theorem}
\begin{proof}
We first consider the value of the LHS. By the rotational invariance of the Gaussian
measure, the value is equal for all $z$ and we can restrict ourselves to $z = e_1$.
\begin{eqnarray*}
\min_{\norm{z}_{2} = 1} \left(\int_{a \in \R^n} \abs{\mydot{a}{z}}^{p} \mu(a) da\right)^{\nfrac{1}{p}}
&=& \left(\int_{\R^{n}} \abs{\mydot{a}{e_{1}}}^{p} \mu(a) da\right)^{\nfrac{1}{p}} \\
&=& \left(\int_{\R^{n}} \abs{a_{1}}^{p} \frac{e^{-\nfrac{\norm{a}^{2}}{2}}}{(2\pi)^{\nfrac{d}{2}}} da_{1} \cdot da_{2} \cdots \cdot da_{n}\right)^{\nfrac{1}{p}} \\
&=& \left(\int_{-\infty}^{\infty} \abs{a_{1}}^{p} \frac{e^{-\nfrac{a_{1}^{2}}{2}}}{\sqrt{2\pi}} da_{1} \cdot \prod_{j=2}^{n} \left(\frac{1}{\sqrt{2\pi}} \int_{-\infty}^{\infty} e^{-\nfrac{a_{j}^{2}}{2}} da_{j}\right)\right)^{\nfrac{1}{p}} \\
&=& \gamma_{p}.
\end{eqnarray*}
In comparison, the optimum of the convex relaxation can be upper bounded by
using the matrix $X = \nfrac1n\cdot  I$.

\begin{eqnarray*}
\min_{\substack{I \succcurlyeq X \succcurlyeq 0 \\ \tr(X)=1}} \left(\int_{a \in \R^n} \left(a^{T}Xa\right)^{\nfrac{p}{2}} \mu(a) da\right)^{\nfrac{1}{p}}
&\leq& \frac{1}{\sqrt{n}} \left(\int_{\R^{n}} \norm{a}^{p} \frac{e^{-\nfrac{\norm{a}^{2}}{2}}}{(2\pi)^{\nfrac{n}{2}}} da_{1} \cdot da_{2} \cdots \cdot da_{n}\right)^{\nfrac{1}{p}}\\
&=& \frac{1}{\sqrt{n}} \left(\int_{\omega \in \sph^{n-1}} \int_{r=0}^{\infty} r^{p} \frac{e^{-\nfrac{r^{2}}{2}}}{(2\pi)^{\nfrac{n}{2}}} r^{n-1} dr\cdot  d\omega\right)^{\nfrac{1}{p}}\\
&=& \frac{1}{\sqrt{n}} \left(\frac{1}{(2\pi)^{\nfrac{(n-1)}{2}}} \int_{0}^{\infty} r^{n+p-1} \frac{e^{-\nfrac{r^{2}}{2}}}{\sqrt{2\pi}} dr \cdot  \int_{\omega \in \sph^{n-1}} d\omega\right)^{\nfrac{1}{p}} \\
&=& \frac{1}{\sqrt{n}} \left(\frac{1}{(2\pi)^{\nfrac{(n-1)}{2}}} \cdot \frac{2^{(n+p-1)/2}  \Gamma(\frac{n+p}{2})}{2\sqrt{\pi}} \cdot \frac{2 \pi^{\nfrac{n}{2}}}{\Gamma(\nfrac{n}{2})}\right)^{\nfrac{1}{p}}\\
&=& \inparen{\inparen{\frac{2}{n}}^{\nfrac{p}{2}} \cdot \frac{\Gamma(\frac{n+p}{2})}{\Gamma(\nfrac{n}{2})}}^{\nfrac{1}{p}}
~\leq~ \inparen{1+\frac{O(p)}{n}}^{\nfrac{1}{2}}
\end{eqnarray*}
where the third equality used that
$\int_{\omega \in \sph^{n-1}} d\omega = \text{area}(\sph^{n-1}) = \frac{2\pi^{\nfrac{n}{2}}}{\Gamma(\nfrac{n}{2})}$,
and $\int_0^{\infty} r^{n+p-1} \frac{e^{-\nfrac{r^2}{2}}}{\sqrt{2\pi}}dr = \gamma_{n+p-1}^{n+p-1}$/2.
Choosing $n \gg p/\eta$ then proves the claim.
\end{proof}

\subsection{Discretizing the gap example}
A discrete analog of the above, i.e., picking sufficiently many samples from the same distribution, gives us our final integrality gap (or ``rank gap'') example.
\begin{theorem} \label{thm:discrete-gap}
Given any $\eta > 0$, there exist $m_{0}, n_{0} \in \Z$ such that for all $m \geq m_{0}$ and $n \geq n_{0}$, if we pick i.i.d. random points $a_{1}, a_{2}, \dotsc, a_{m} \in \R^{n}$ with each point having i.i.d. $N(0,1)$ coordinates, then with some non-zero probability,
\[
\min_{\norm{z}_{2}=1} \left(\frac{1}{m} \sum_{i=1}^{m} \abs{\mydot{a_{i}}{z}}^{p}\right)^{\nfrac{1}{p}} \geq (1-\eta) \cdot \gamma_{p} \cdot \min_{\substack{I \succcurlyeq X \succcurlyeq 0 \\ \tr(X)=1}} \left(\int_{a \in \R^n} \abs{a_{i}^{T} X a_{i}}^{\nfrac{p}{2}}\right)^{\nfrac{1}{p}}.
\]
In other words, there exist points $b_{1}, b_{2}, \dotsc, b_{m} \in \R^{n}$, where $b_{i} \defeq  m^{-\nfrac{1}{p}} \cdot a_{i}$, giving the desired integrality gap example.
\end{theorem}

The theorem can be proved by using the continuous gap instance, and concentration bounds for 
the samples $a_1, \ldots, a_m$. We defer a full proof to the appendix.

%%%%%%%%%%%%%%%% UG-HARDNESS %%%%%%%%%%%%%%%%%%%%%%
\section{Unique-Games Hardness}\label{sec:ug-hard}
\subsection{Khot's Unique Games Conjecture}
%\Nnote{There is some strange space after \ULC I think it is misbehaving in the math mode $\ULC$}
We shall show a reduction to subspace approximation problem from the Unique Label Cover problem
defined below.
\begin{definition}
An instance of \ULC with alphabet size $R$ is specified as a bipartite graph $\instance$ with a
set of permutations $\{\pi_{vw}: [R] \to [R]\}_{(v,w) \in E}$. A labeling  $\L: V \cup W \to [R]$
is said to satisfy an edge $(v,w)$ if $\L(w) = \pi_{vw}(\L(v))$. We denote by $\val(\U)$ the
maximum fraction of edges satisfied by any labeling $\L$.
\end{definition}
The \UGC proposed by Khot in \cite{Khot02} conjectures the hardness of distinguishing between the
cases when the optimum to the above problem is very close to 1 and when it is very close to 0.
This conjecture is an important complexity assumption as several approximation
problems have been shown to be at least as hard as deciding if a given
instance $\U$ of \ULC problem has $\val(\U) > 1-\epsilon$ or $\val(\U) < \delta$ for
appropriate positive constants $\epsilon$ and $\delta$.
\begin{conjecture}[Khot \cite{Khot02}]
Given any constants $\epsilon, \delta > 0$, there is an integer $R$ such that
it is NP-hard to decide if for given an instance
$\instance$ of Unique Label Cover with alphabet size $R$,
$\val(\U) \geq  1-\epsilon$ or $\val(\U) \leq  \delta$.
\end{conjecture}

\subsection{Reduction from Unique Label Cover}
We will now prove Unique-Games hardness of approximating \subspace{$n-1$}{$p$}
within a factor better than $\gamma_p$.  As in  Section \ref{sec:gap}, this also gives
a hardness approximating \subspace{$k$}{$p$} for $k$ which is a sufficiently large
function of $k$, by a trivial embedding of the given instance $\R^n$ into $\R^{n'}$ such
$k(n') = n-1$. If we want $n'$ to be a polynomial in $n$, this will give a hardness for all
$k = n^{\Omega(1)}$.

We describe below the reduction from an instance \instance  of \ULC with alphabet size $R$
to \subspace{$n-1$}{$p$}.
The variables in our reduction will be of the form $b_{w,i}$ for each $w \in W$ and
$i \in [R]$. We denote the vector $(b_{w,1}, \ldots, b_{w,R})$ by
$\bfb_w$ and for each $v \in V$, define $\bfb_v \defeq \av_{w \in N(v)}[\pi_{wv}(\bfb_w)]$.
For any $\bfb \in \R^R$, we define the function $f_b: \pmone^R \to \R$ as
\[ f_{b}(x_1, \ldots, x_R) ~\defeq~ \sum_{i=1}^R x_i \cdot b_{i}  \]
Norms for functions are defined as usual (over the uniform probability measure).
Note that $\norm{f_{b}}_2^2 = \norm{\bfb}_2^2$. When the exponent in the norm is
unspecified, $\norm{\cdot}$ denotes $\norm{\cdot}_2$.

Given an instance $\instance$ of \ULC we output the following instance of subspace
approximation, for a suitable constant $B$ to be determined later:
\begin{align*}
\text{minimize} &~~~~~\av_{(v,w) \in E}\left[\norm{\fb{v}}_p^p\right]
 + B \cdot \av_{(v,w) \in E}\left[\norm{\fb{v} - f_{\pi_{wv}(b_w)}}_p^p\right]\\
\text{subject to:} &~~~~~\av_{(v,w) \in E}\left[\norm{\fb{w}}_2^2\right]
~=~ \av_{(v,w) \in E}\left[ \norm{\bfb_w}_2^2 \right] ~\geq~  1
\end{align*}
Note that the variables in the problem are only the vectors $\bfb_w$ for all $w \in W$.
It is easy to verify the functions $\fb{v}$ and $\fb{w}$ can be generated by application
of an appropriate operator $A$. In the proof below we shall often drop the subscript
on the permutations $\pi_{wv}$ when it is clear from the context.
Note that value of instance of \subspace{$n-1$}{$p$} is actually the $p^{th}$
root of the above objective. Let $(\opt)^p$ denote the
optimal value for the above objective (so that $\opt$ is the optimal value for
\subspace{$n-1$}{$p$}).
%\Nnote{I observed that the command for probabilities $\pr$ and ${\sf Pr}$ are not the same. }

\subsubsection*{Completeness}
The following claim shows that the optimum of the subspace approximation
problem is low when the \ULC is instance is highly satisfiable.
\begin{claim}\label{clm:ug-opt}
If $\val(\U) \geq 1-\epsilon$, then $(\opt)^p \leq 1 + \epsilon \cdot  B \cdot 2^p.$
\end{claim}
\begin{proof}
By assumption, there exists a labeling $\L: V \cup W \to [R]$ such that
$\pr_{(v,w) \in E}[\L(v) \neq \pi_{wv}(\L(w))] \leq \epsilon$. We
construct a solution the above instance of the subspace approximation
problem, taking $b_{w,i} = 1$ if $\L(w) = i$ and 0 otherwise. It is easy
to check that $\av_{(v,w) \in E}\left[\norm{\fb{w}}_2^2\right] = 1$.

We now bound the value of the objective function. First note that
$\fb{v} = \av_{w \in N(v)}[\fbpi{w}]$ is bounded between -1 and 1,
which implies $\av_{(v,w)\in E}\left[\norm{\fb{v}}_p^p\right] \leq 1$.
To bound the second term, we can use Jensen's Inequality to get
\begin{eqnarray*}
 \av_{(v,w) \in E}\left[\norm{\fb{v} - f_{\pi_{wv}(b_w)}}_p^p\right]
&=& \av_{(v,w_1) \in E}\left[\norm{\av_{w_2 \in N(v)}\left[f_{\pi_{w_2v}(b_{w_2})}\right] - f_{\pi_{w_1v}(b_{w_1})}}_p^p\right] \\
&\leq&  \av_{v,w_1,w_2}\left[\norm{f_{\pi_{w_2v}(b_{w_2})} - f_{\pi_{w_1v}(b_{w_1})}}_p^p\right].
\end{eqnarray*}
Note that $\norm{f_{\pi_{w_2v}(b_{w_2})} - f_{\pi_{w_1v}(b_{w_1})}}_p^p$ equals
$2^{p-1}$ if $\pi_{w_1v}(\L(w_1)) \neq \pi_{w_2v}(\L(w_2))$ and 0 otherwise.
Hence,
\begin{eqnarray*}
\av_{v,w_1,w_2}\left[\norm{f_{\pi_{w_2v}(b_{w_2})} - f_{\pi_{w_1v}(b_{w_1})}}_p^p\right]
&=& 2^{p-1} \cdot \Pr_{v,w_1,w_2}\left[\pi_{w_1v}(\L(w_1)) \neq \pi_{w_2v}(\L(w_2))\right] \\
&\leq& 2^{p-1} \left(\Pr_{v,w_1}\left[\pi_{w_1v}(\L(w_1)) \neq \L(v)\right] +
\Pr_{v,w_2}\left[\pi_{w_2v}(\L(w_2)) \neq \L(v)\right] \right) \\
&\leq& 2^p\cdot \epsilon.
\end{eqnarray*}
Combining the two bounds above gives $(\opt)^p \leq 1 + \epsilon \cdot B \cdot 2^p$.
\end{proof}

\subsubsection*{Soundness}
For the soundness, we need to prove that if $\val(\U) \leq \delta$, then
$\opt \geq \gamma_p^p \cdot (1-\nu)$ where $\nu$ is a small constant depending on
$\epsilon$ and $\delta$. We first make some simple observations about the optimal solution.
\begin{claim}\label{clm:ug-opt}
For any optimal solution $\{\bfb_w\}_{w \in W}$ to the above instance of \subspace{$n-1$}{$p$},
it must be true that
\begin{enumerate}
\item $\av_{(v,w) \in E} \left[ \norm{\bfb_v}^2 \right]  ~\leq~ \av_{(v,w) \in E} \left[
    \norm{\bfb_w}^2 \right]  ~=~ 1$
\item $\av_{(v,w) \in E}\left[\norm{\fb{v} - \fbpi{w}}_p^p\right]  \leq \nfrac{\gamma_p^p}{B}$.
\end{enumerate}
\end{claim}
\begin{proof}
Since scaling all vectors by a constant less than 1 can only improve the value of
the objective, we can assume that for the vectors $\{\bfb_w\}_{w \in W}$
in the solution $\av_{(v,w) \in E} \insquar{\norm{\bfb_w}_2^2} = 1$. Then
Jensen's inequality gives\\
\begin{equation*}
\av_{(v,w) \in E} \left[ \norm{\bfb_v}^2 \right]
~=~ \av_{(v,w) \in E} \left[ \norm{\av_{w'\in N(v)} \insquar{\pi_{w'v}(\bfb_{w'})}}^2 \right]
~\leq~ \av_{(v,w) \in E}\insquar{\norm{\bfb_w}^2} ~=~ 1.
\end{equation*}

To deduce the second fact, we show that there exists a feasible solution $\{\bfb_{w}\}_{w \in W}$ such that
$\opt \leq \gamma_p^p$. For all $w \in W$, we take $\bfb_w = (\nfrac{1}{\sqrt{R}}, \ldots, \nfrac{1}{\sqrt{R}})$.
The solution is feasible since $\norm{\bfb_w} = 1$ for each $w \in W$ and also
$\av_{(v,w) \in E}\left[\norm{\fb{v} - \fbpi{w}}_p^p\right] = 0$.
Also, since $\fb{v}$ is a linear function of Bernoulli variables and $\norm{\bfb_v} = 1$,
Claim \ref{clm:upper-gaussian} gives that for each $v \in V$, $\norm{\fb{v}}_p \leq \gamma_p$.
\end{proof}

We show that if $\val(\U) \leq \delta$, then in fact the first term itself is
approximately $\gamma_p^p$.
%In the following for any random variable $Z(v)$, we use
%$\av_{v \in V}Z(v)$ to denote $\av_{(v,w) \in E} Z(v)$ i.e. the measure of each
%$v \in V$ in the expectation is proportional to its degree.
As is standard in Unique Games based reductions, the
proof proceeds by arguing separately about the ``high-influence'' and
``low-influence'' cases. However, since the inputs for our problem
are not in the form of a long-code but the vectors $\bfb$, we will use
$\max_{i \in R} \{\abs{b_{i}}/\norm{\bfb}\}$ as a substitute for
influence of the $i^{th}$  variable on the function $\fb{}$.

%\begin{lemma}
%There exists a small constant $\nu=\nu(\epsilon,\delta,B)$ such that if
%$\val(\U) \leq \delta$, then
%$$\av_{(v,w) \in E} \left[\norm{\fb{v}}_p^p\right] \geq \gamma_p^p(1-\nu)$$
%\end{lemma}
For the vertices $v \in V$ where the functions $\fb{v}$ have no influential
coordinates, the Central Limit Theorem shows that $\norm{\fb{v}}_p$ is very
close to $\gamma_p$. We then show that the contribution of the remaining
vertices to the objective function is small.

Below, we define $S_1$ to be the set of vertices corresponding to low influence
functions and divide the remaining vertices into three cases which we shall
analyze separately. The parameters $\tau, \beta \in (0,1/2)$ will be chosen later.
\begin{eqnarray*}
S_1 &\defeq & \left\{ v \in V ~~\left\lvert~~
  \max_{i \in [R]}\{\abs{b_{v,i}}\} < \tau\cdot \norm{\bfb_v}
\right.\right\}\\
S_2 & \defeq & \left\{ v \in V ~~\left\lvert~~
  \norm{\bfb_v}^2  \leq (1-\beta) \cdot \av_{w \in N(v)}\left[\norm{\bfb_w}^2\right]
\right.\right\}\\
S_3 & \defeq & \left\{ v \in V \setminus S_2 ~~\left\lvert~~
  \exists i~ s.t. ~~\abs{b_{v,i}} \geq \tau \cdot  \norm{\bfb_v} ~\text{and}~
  \Pr_{w \in N(v)}\left[\abs{b_{w,\pi_{vw}(i)}} \geq \nfrac{\tau}{4}\cdot \norm{\bfb_w} \right] \leq \nfrac14
\right.\right\}\\
S_4 &\defeq & V \setminus (S_1 \cup S_2 \cup S_3).
\end{eqnarray*}

Since $\fb{v}(x_1, \ldots, x_R) = b_{v,1} \cdot x_1 + \cdots + b_{v,R} \cdot x_R$ is a linear
function of Bernoulli variables, Claim \ref{clm:clt} gives that
\begin{equation}\label{eqn:S1}
\forall v \in S_1 \qquad
\norm{\fb{v}}_p^p ~~\geq~~ \gamma_p^p \cdot \norm{\bfb_v}_2^p \cdot  \inparen{1 - 10\tau\cdot(\log(\nfrac{1}{\tau}))^{\nfrac{p}{2}}}
\end{equation}

Note that the norm $\norm{\fb{v}}_p$  may be unbounded for individual vertices. Hence we will use the quantity
$\av_{(v,w) \in E} \insquar{\ind{S_i}(v) \cdot \norm{\bfb_v}^2} $ as a measure of the contribution of the set
$S_i$ to the objective, where $\ind{S_i}(\cdot)$ is the indicator function of the set $S_i$. Claims
\ref{clm:S2}, \ref{clm:S3} and \ref{clm:S4} help bound the contribution of the sets $S_2, S_3$
and $S_4$.

\begin{claim}\label{clm:S2}
\[\av_{(v,w) \in E}\left[(1-\ind{S_2}(v)) \cdot \norm{\bfb_v}^2 \right] \geq 1- \beta -
\frac{4\gamma_p^2}{\beta B^{\nfrac{2}{p}}}.\]
\end{claim}
\begin{proof}
Since $\bfb_w = \av_{w \in N(v)}[\pi_{wv}(\bfb_{b_w})]$, being in $S_2$ means that on average, many
vectors $\bfb_w$ differ from $\bfb_v$. We use this to get a bound on the measure of $S_2$. We have
\begin{eqnarray*}
\norm{\bfb_v}^2  \leq (1-\beta) \cdot \hspace{-.2cm} \av_{w \in N(v)}\left[\norm{\bfb_w}^2\right] 
&\implies& 
\beta \cdot \hspace{-.2cm} \av_{w \in N(v)}\left[\norm{\bfb_w}^2\right]  \leq \hspace{-.2cm} \av_{w \in N(v)}\left[\norm{\bfb_w}^2\right] -\norm{\bfb_v}^2 \\
&\implies& 
\beta \cdot \hspace{-.2cm} \av_{w \in N(v)}\left[\norm{\bfb_w}^2\right]  \leq \hspace{-.2cm} \av_{w \in N(v)}\left[\norm{\pi_{wv}(\bfb_w) - \bfb_v}^2 \right],
\end{eqnarray*}
as $\norm{\bfb_w} = \norm{\pi_{wv}(\bfb_{w})}$ and that
$\bfb_v$ is the mean of $\pi_{wv}(\bfb_{w})$. Now, since $\norm{\bfb} = \norm{\fb{}}$,
we get that
\begin{eqnarray*}
\beta \cdot \hspace{-.2cm} \cdot \av_{(v,w) \in E}\left[ \ind{S_2}(v) \cdot \norm{\bfb_{w}}^2 \right]
&\leq& \av_{(v,w) \in E}\left[ \norm{\fb{v} - \fbpi{w}}_2^2\right] \\
&\leq& \av_{(v,w) \in E}\left[ \norm{\fb{v} - \fbpi{w}}_p^2\right] \qquad \quad \quad (\text{since}~ \norm{f}_2 \leq \norm{f}_p) \\
&\leq& \left(\av_{(v,w) \in E}\left[ \norm{\fb{v} - \fbpi{w}}_p^p\right] \right)^{\nfrac{2}{p}} ~\quad \text{(using Jensen's Inequality)}\\
&\leq& \nfrac{\gamma_p^2}{B^{\nfrac{2}{p}}}
\end{eqnarray*}
where we used the assumption that
$\av_{(v,w) \in E}\left[ \norm{\fb{v} - \fbpi{w}}_p^p\right] \leq \nfrac{\gamma_p^p}{B}$.
This gives that
\begin{align*}
\av_{(v,w) \in E}\left[(1-\ind{S_2}(v)) \cdot \norm{\bfb_v}^2 \right] 
& \geq (1-\beta) \cdot \av_{(v,w) \in E}\left[(1-\ind{S_2}(v)) \cdot \norm{\bfb_w}^2 \right] \\
& \geq (1-\beta) \cdot \left(1 - \frac{\gamma_p^2}{\beta B^{\nfrac{2}{p}}}\right) \\
& \geq  1 - \beta - \frac{\gamma_p^2}{\beta B^{\nfrac{2}{p}}}.
\end{align*}
The second inequality above used that $\av_{(v,w) \in E}\left[ \norm{\bfb_w}^2 \right] = 1$ from
claim \ref{clm:ug-opt}.
\end{proof}

\begin{claim}\label{clm:S3}
$\av_{(v,w) \in E}\left[\ind{S_3}(v) \cdot \norm{\bfb_v}^2 \right]
~\leq~ \nfrac{16}{\tau^2} \cdot \nfrac{\gamma_p^2}{B^{\nfrac{2}{p}}}$.
\end{claim}
\begin{proof}
Consider a vertex $v \in S_3$. Since we know that $v \notin S_2$, we get that
\begin{equation*}
\Pr_{w \in N(v)}\left[\norm{\bfb_w} \geq 2\norm{\bfb_v} \right]
  ~~\leq~~ \frac{\av_{w \in N(v)}\left[\norm{\bfb_w}^2\right]}{4\norm{\bfb_v}^2}
  ~~\leq~~ \frac{1}{4-4\beta}.
\end{equation*}

Fix and $i \in [R]$ such that $\abs{b_{v,i}} \geq \tau \cdot  \norm{\bfb_v}$ and
$\Pr_{w \in N(v)}\left[\abs{b_{w,\pi_{vw}(i)}} \geq \nfrac{\tau}{4}\cdot \norm{\bfb_w} \right] \leq \nfrac14$.
By a union bound,
\begin{equation*}
\Pr_{w \in N(v)}\left[\norm{\bfb_w} \leq 2\norm{\bfb_v} ~~\text{and}~~
  \abs{b_{w,\pi_{vw}(i)}} \leq \nfrac{\tau}{4}\cdot \norm{\bfb_w} \right]
  ~~\geq~~ 1 - \nfrac14 - \nfrac{1}{(4-4\beta)} ~~>~~ \nfrac14.
\end{equation*}

Using this we can again say that $\norm{\bfb_{v} - \pi_{wv}(\bfb_w)}$ must be large on
average and, hence, derive a bound on the measure of $S_3$.
\begin{eqnarray*}
\av_{w \in N(v)}\left[ \norm{\bfb_v - \pi_{wv}(\bfb_w)}^2\right]
&\geq& \av_{w \in N(v)} \left[\abs{b_{v,i} - b_{w,\pi(i)}}^2\right] \\
&\geq& \nfrac14 \cdot \abs{\tau\cdot  \norm{\bfb_v} - \nfrac{\tau}{4} \cdot 2\norm{\bfb_v}}^2 \\
&\geq& \nfrac{\tau^2}{16}\cdot \norm{\bfb_v}^2.
\end{eqnarray*}

As in the previous claim, we use this to conclude that
\begin{eqnarray*}
\av_{(v,w) \in E}\left[\ind{S_3}(v) \cdot \norm{\bfb_v}^2 \right]
&\leq& \nfrac{16}{\tau^2} \cdot \av_{(v,w) \in E}\left[ \norm{\fb{v} - \fbpi{w}}_2^2 \right] \\
&\leq& \nfrac{16}{\tau^2} \left(\av_{(v,w) \in E}\left[ \norm{\fb{v} - \fbpi{w}}_p^p \right]\right)^{\nfrac{2}{p}}
~~\leq~~ \nfrac{16}{\tau^2} \cdot \nfrac{\gamma_p^2}{B^{\nfrac{2}{p}}}.
\end{eqnarray*}
\end{proof}

\begin{claim}\label{clm:S4}
$\av_{(v,w) \in E}\left[\ind{S_4}(v) \right] \leq \nfrac{64\delta}{\tau^2}$.
\end{claim}
\begin{proof}
Since $v \notin S_1 \cup S_2 \cup S_3$, we know that
\begin{equation*}
\exists i \in [R] ~~~\text{such that}~~
\Pr_{w \in N(v)} \left[ \abs{b_{w, \pi_{vw}(i)}} \geq \nfrac{\tau}{4}\cdot \norm{\bfb_w}\right] \geq \nfrac14.
\end{equation*}

Construct a labeling for $\U$ by assigning to each $v \in V$, the special label $i$ as above,
and to each $w \in W$, a random label $j$ satisfying $\abs{b_{w,j}} \geq \nfrac{\tau}{4}\cdot \norm{\bfb_w}$.
For, $w \in W$ when no such $j$ exists or for $v \notin S_4$, we fix a label arbitrarily.

Note that there can be at most $\nfrac{16}{\tau^2}$ choices of $j$ satisfying
$\abs{b_{w,j}} \geq \nfrac{\tau}{4}\cdot \norm{\bfb_w}$. By the condition on $i$, we know that,
in expectation, the labeling satisfies $\nfrac14 \cdot \nfrac{\tau^2}{16}$ fraction of
the edges incident on a $v \in S_4$. Since the fraction of edges satisfied overall is at most
$\delta$, we get that
\begin{equation*}
\av_{(v,w) \in E}\left[\ind{S_4}(v) \cdot \nfrac{\tau^2}{64} \right] \leq \delta
~~\implies~~ \av_{(v,w) \in E}\left[\ind{S_4}(v) \right] \leq \nfrac{64\delta}{\tau^2}.
\end{equation*}
\end{proof}

Let $\nu$ denote $10\tau \cdot \inparen{\log(\nfrac{1}{\tau})}^{\nfrac{p}{2}}$.
Using these estimates, we can now prove the soundness of the reduction.

\begin{lemma} \label{lemma:sound}
If $\val(\U) < \delta$, then for the reduction with parameters $B, \tau$ and $\beta = \tau^2$
\begin{equation*}
(\opt)^p ~\geq~
\gamma_p^p \cdot \left(
1 ~-~ \nu - \frac{p\tau^2}{2} 
~-~  \frac{10 p\cdot \gamma_p^2}{\tau^2 B^{\nfrac{2}{p}}} 
~-~ \frac{p\gamma_p^2}{2}\left(\frac{64\delta}{\tau^2}\right)^{\nfrac{(p-2)}{p}}
\right)
\end{equation*}
\end{lemma}

\begin{proof}
% Let $\nu$ denote $10\tau \cdot \inparen{\log(\nfrac{1}{\tau})}^{\nfrac{p}{2}}$. 
Using (\ref{eqn:S1}) we have that
\begin{equation*}
(\opt)^p 
~\geq~ \av_{(v,w) \in E} \left[ \ind{S_1}(v) \cdot \gamma_p^p \cdot (1-\nu) \norm{\bfb_v}_2^p\right]
~\geq~ \gamma_p^p \cdot (1-\nu) \cdot \left( \av_{(v,w) \in E} \left[ \ind{S_1}(v) \cdot \norm{\bfb_v}_2^2\right] \right)^{p/2}.
\end{equation*}

We lower bound $\ind{S_1}(v)$ by
$1 - \ind{S_2}(v) - \ind{S_3}(v) - \ind{S_4}(v)$. Claims \ref{clm:S2} and \ref{clm:S3}
and give bounds on the first two terms (with $\beta = \tau^2$).
\begin{align*}
\av_{(v,w) \in E}\left[(1-\ind{S_2}(v)) \cdot \norm{\bfb_v}^2 \right] & 
~~\geq~~ 1- \tau^2 - \frac{4\gamma_p^2}{\tau^2 B^{\nfrac{2}{p}}},\\
\av_{(v,w) \in E}\left[ \ind{S_3} \norm{\bfb_v}^2\right] & 
~~\leq~~ \frac{16  \gamma_p^2}{\tau^2 B^{\nfrac{2}{p}}}
\end{align*}

% \begin{eqnarray*}
% \av_{(v,w) \in E}\left[ \ind{S_2} \norm{\bfb_v}\right]
%   &\leq& \left(\av_{(v,w) \in E}\left[ \ind{S_2} \norm{\bfb_v}^2\right]\right)^{\nfrac{1}{2}}
%   ~\leq~ \frac{2\gamma_p}{B^{\nfrac{1}{p}}}\\
% \av_{(v,w) \in E}\left[ \ind{S_3} \norm{\bfb_v}\right]
%   &\leq& \left(\av_{(v,w) \in E}\left[ \ind{S_3} \norm{\bfb_v}^2\right]\right)^{\nfrac{1}{2}}
%   ~\leq~ \frac{4\gamma_p}{\tau B^{\nfrac{1}{p}}}\\
% \av_{(v,w) \in E}\left[ \ind{S_4} \norm{\bfb_v}\right]
%   &\leq& \left(\av_{(v,w) \in E}\left[ \ind{S_4}(v)\right]\right)^{\nfrac{1}{2}}
%          \left(\av_{(v,w) \in E}\left[ \norm{\bfb_v}^2\right]\right)^{\nfrac{1}{2}}
%   ~\leq~ \frac{8\sqrt{\delta}}{\tau}
% \end{eqnarray*}
% where the last bound used the assumption from Claim \ref{clm:ug-opt} that
% $\av_{(v,w) \in E}\left[\norm{\bfb_v}^2\right] \leq 1$.
We bound the third term using Claim \ref{clm:S4} and H\"{o}lder's inequality
\begin{align*}
\av_{(v,w) \in E}\left[ \ind{S_4} \norm{\bfb_v}^2\right] 
& ~\leq~ \left(\av_{(v,w) \in E}\left[ \ind{S_4}(v)\right]\right)^{\nfrac{(p-2)}{p}} \left(\av_{(v,w) \in E}\left[ \norm{\bfb_v}^p\right]\right)^{\nfrac{2}{p}} \\
& ~\leq~ \inparen{\frac{64\delta}{\tau^2}}^{\nfrac{(p-2)}{p}} \cdot \gamma_p^2,
\end{align*}
where the last bound used that since $\opt \leq \gamma_p$ (see Claim \ref{clm:ug-opt}), we must have
\begin{align*}
\av_{(v,w) \in E}\left[ \norm{\bfb_v}_2^p\right] 
~=~ \av_{(v,w) \in E}\left[\norm{f_{b_v}}_2^p\right]
~\leq~ \av_{(v,w) \in E}\left[\norm{f_{b_v}}_p^p\right] 
~\leq~ \gamma_p^p .
\end{align*}

Combining the bounds for the above three terms proves the lemma.
% Combining this with the above bounds gives
% \begin{equation*}
% (\opt)^p ~\geq~ \gamma_p^p \cdot (1-\nu) \cdot \left( 1- \frac{2\gamma_p}{B^{\nfrac{1}{p}}}
% - \frac{4\gamma_p}{\tau B^{\nfrac{1}{p}}} - \frac{8\sqrt{\delta}}{\tau} \right)^p
% ~\geq~ \gamma_p^p\cdot  \inparen{1- 10\tau \inparen{\log(\nfrac{1}{\tau})}^{\nfrac{p}{2}}
% - \frac{6p\gamma_p}{\tau B^{\nfrac{1}{p}}} - \frac{8p\sqrt{\delta}}{\tau}}
% \end{equation*}
% which proves the lemma.
\end{proof}

For a small constant $\eta$ such that $\eta (\log(\nfrac{1}{\eta}))^{\nfrac{p}{2}} < \nfrac{2^{-\nfrac{p}{2}}}{50}$, choosing parameters as
\begin{align*}
\tau \defeq  \eta^2/p, & \qquad\delta \defeq \inparen{\frac{\eta}{p\gamma_p^2}}^{\nfrac{p}{(p-2)}} \cdot \frac{\tau^2}{64},\\
B \defeq \inparen{\frac{40p\gamma_p^2}{\eta \cdot \tau^2}}^p,   & \qquad \text{and} \quad \epsilon \defeq \frac{\eta}{2^p \cdot B}\\
\end{align*}
in Lemma \ref{lemma:sound} would imply that $\opt \leq 1+\eta$ in the completeness case and $\opt
\geq \gamma_p \cdot (1-\eta)$ in the soundness case. This gives the following theorem.
\begin{theorem}
For any $p \geq 2$ and sufficiently small constant $\eta$, there exist constants $\epsilon, \delta > 0$
and a reduction from \ULC to \subspace{$n-1$}{$p$} such that if $\val(\U)$ is
the fraction of edges satisfiable in the given instance of \ULC and $\opt$ is the optimum of the
instance of \subspace{$n-1$}{$p$}, then
\begin{align*}
\val(\U) \geq 1-\epsilon & \implies \opt \leq 1 + \eta \quad \text{and} \\
\val(\U) \leq \delta & \implies \opt \geq \gamma_p \cdot (1-\eta).
\end{align*}
\end{theorem}
% For a small constant $\eta$ such that $\eta (\log(\nfrac{1}{\eta}))^{\nfrac{p}{2}} < \nfrac{2^{-\nfrac{p}{2}}}{30}$, choosing parameters as
% \[ \epsilon \defeq \frac{\eta^{3p+1}}{p(40p\gamma_p)^p} \quad , \quad
%   \delta \defeq  \frac{\eta^6}{600 p^2} \quad , \quad
%   \tau \defeq  \eta^2 \quad  \text{and} \quad
%   B \defeq \inparen{\frac{20p\gamma_p}{\eta^3}}^p
% \]
% in Lemma \ref{lemma:sound} would imply that $\opt \leq 1+\eta$ in the completeness case and $\opt \geq \gamma_p \cdot (1-\eta)$ in the soundness case. This gives the following theorem.
% \begin{theorem}
% For any $p \geq 2$ and sufficiently small constant $\eta$, there exist constants $\epsilon, \delta > 0$
% and a reduction from \ULC to \subspace{$n-1$}{$p$} such that if $\val(\U)$ is
% the fraction of edges satisfiable in the given instance of \ULC and $\opt$ if the optimum of the
% instance of \subspace{$n-1$}{$p$}, then
% \[\val(\U) \geq 1-\epsilon \implies \opt \leq 1 + \eta
% \quad \text{and} \quad
% \val(\U) \leq \delta \implies \opt \geq \gamma_p \cdot (1-\eta)\]
% \end{theorem}

\section*{Acknowledgments}
We thank Kasturi Varadarajan for initiating the work on this problem by suggesting that
we generalize the algorithm of \cite{VVYZ} and for generous
help with an early draft of this paper on which he was offered a co-authorship
(but later opted out). MT would also like to thank David Steurer for helpful discussions.
We thank the anonymous reviewers of this manuscript for their suggestions and references,
and for pointing out an error in the previous proof of Lemma \ref{lemma:sound}.

\bibliography{subspace-ref}

%%%%%%%%%%%%%%%%%%%%%%%%%%%%%%%%%%%%%%%%%%%%%%%%%%%%%%%%%%%%%%%%
%%%%%%% APPENDIX %%%%%%%%%%%%%%%%%%%%%%%%%%%%%%%%%%%%%%%%%%%%%%%
%%%%%%%%%%%%%%%%%%%%%%%%%%%%%%%%%%%%%%%%%%%%%%%%%%%%%%%%%%%%%%%%

\appendix

\section{Proof of Theorem \ref{thm:discrete-gap}}
\renewcommand{\thetheorem}{A.\arabic{theorem}}
\setcounter{theorem}{0}  % reset counter

We restate the theorem below.

\begin{theorem} \label{thm:discrete-gap-appendix}
Given any $\eta > 0$, there exist $m_{0}, n_{0} \in \Z$ such that for all $m \geq m_{0}$ and $n \geq n_{0}$, if we pick i.i.d. random points $a_{1}, a_{2}, \dotsc, a_{m} \in \R^{n}$ with each point having i.i.d. $N(0,1)$ coordinates, then with some non-zero probability,
\[
\min_{\norm{z}_{2}=1} \left(\frac{1}{m} \sum_{i=1}^{m} \abs{\mydot{a_{i}}{z}}^{p}\right)^{\nfrac{1}{p}} \geq (1-\eta) \cdot \gamma_{p} \cdot \min_{\substack{I \succcurlyeq X \succcurlyeq 0 \\ \tr(X)=1}} \left(\int_{a \in \R^n} \abs{a_{i}^{T} X a_{i}}^{\nfrac{p}{2}}\right)^{\nfrac{1}{p}}.
\]
In other words, there exist points $b_{1}, b_{2}, \dotsc, b_{m} \in \R^{n}$, where $b_{i} \defeq  m^{-\nfrac{1}{p}} \cdot a_{i}$, giving the desired integrality gap example.
\end{theorem}
\begin{proof}
Let $a_{1}, a_{2}, \dotsc, a_{m}$ be i.i.d. random points in $\R^{n}$, where each point has i.i.d. $N(0, 1)$ coordinates. Then, as we have seen above
\begin{align*}
\expec{\abs{\mydot{a_{i}}{y}}^{p}} & = \int_{\R^{n}} \abs{\mydot{a}{y}}^{p} \mu(a) da = \gamma_{p}^{p}, & \qquad \text{for $y \in \sph^{n-1}$}, \\
\var{\abs{\mydot{a_{i}}{y}}^{p}} & = \expec{\abs{\mydot{a_{i}}{y}}^{2p}} - \expec{\abs{\mydot{a_{i}}{y}}^{p}}^{2} = \gamma_{2p}^{2p} - \gamma_{p}^{2p}, & \qquad \text{for $y \in \sph^{n-1}$}.
\end{align*}
By Chebyshev's Inequality,
\[
\prob{\frac{1}{m} \sum_{i=1}^{m} \abs{\mydot{a_{i}}{y}}^{p} \leq (1-\epsilon)\gamma_{p}^{p}} \leq \frac{(\gamma_{2p}^{2p} - \gamma_{p}^{2p})}{m\epsilon^{2}\gamma_{p}^{2p}}.
\]
Let $\mathcal{N}$ be any $\delta$-net of the unit sphere (i.e., $\mathcal{N} \subseteq \sph^{n-1}$ such that for any $z \in \sph^{n-1}$, there exists some $y \in \mathcal{N}$ such that $\norm{y-z}_{2} \leq \delta$), where $\delta$ is a parameter that will be picked later. It is known (e.g. see Claim 2.9 in \cite{FO05}) how to construct such $\delta$-nets of $\sph^{n-1}$ with size as small as $\size{\mathcal{N}} \leq (\frac{9}{\delta})^n$. Now using union bound over $\mathcal{N}$
\[
\prob{\frac{1}{m} \sum_{i=1}^{m} \abs{\mydot{a_{i}}{y}}^{p} \geq (1-\epsilon)\gamma_{p}^{p},~ \text{for all $y \in \mathcal{N}$}} \geq 1 - \frac{(\frac{9}{\delta})^n \cdot (\gamma_{2p}^{2p} - \gamma_{p}^{2p})}{m \epsilon^{2} \gamma_{p}^{2p}} > \frac{3}{4},
\]
as long as we choose $m$ large enough so that
\[
m > \frac{4 \cdot (\frac{9}{\delta})^n \cdot (\gamma_{2p}^{2p} - \gamma_{p}^{2p})}{\epsilon^{2} \gamma_{p}^{2p}}.
\]
For any $z \in \sph^{n-1}$, using $y \in \mathcal{N}$ closest to it
\begin{align*}
\sum_{i=1}^{m} \abs{\mydot{a_{i}}{z}}^{p} & = \sum_{i=1}^{m} \abs{\mydot{a_{i}}{y} + \mydot{a_{i}}{z-y}}^{p} \\
& \geq \sum_{i=1}^{m} \mydot{a_{i}}{y}^{p} - p \delta \sum_{i=1}^{m} \norm{a_{i}}_{2}^{p-1}.
\end{align*}
Therefore,
\[
\prob{\min_{\norm{z}_{2}=1} \frac{1}{m} \sum_{i=1}^{m} \abs{\mydot{a_{i}}{z}}^{p} \geq (1-\epsilon)\gamma_{p}^{p} - \frac{p\delta}{m}\sum_{i=1}^{m} \norm{a_{i}}_{2}^{p-1}} > \frac{3}{4}.
\]
But we also know that
\begin{align*}
\expec{\norm{a_{i}}_{2}^{p-1}} & = \int_{a \in \R^{n}} \norm{a}_{2}^{p-1} \mu(a) da = n^{(p-1)/2} (1+o(1)) \\
\var{\norm{a_{i}}_{2}^{p-1}} & = \expec{\norm{a_{i}}_{2}^{2p-2}} - \expec{\norm{a_{i}}_{2}^{p-1}}^{2} = n^{(p-1)} (1+o(1))
\end{align*}
By Chebyshev's Inequality,
\[
\prob{\frac{1}{m} \sum_{i=1}^{m} \norm{a_{i}}_{2}^{p-1} \geq (1+\epsilon) n^{(p-1)/2} (1+o(1))} \leq \frac{1+o(1)}{m\epsilon^{2}}.
\]
Hence,  choosing $m > \nfrac{5}{\epsilon^{2}}$, we have
\[
\prob{\frac{1}{m} \sum_{i=1}^{m} \norm{a_{i}}_{2}^{p-1} \leq (1+\epsilon) n^{(p-1)/2} (1+o(1))} \geq 1 - \frac{1+o(1)}{m\epsilon^{2}} > \frac{3}{4}.
\]
Putting these together,
\[
\prob{\min_{\norm{z}_{2}=1} \frac{1}{m} \sum_{i=1}^{m} \abs{\mydot{a_{i}}{z}}^{p} \geq (1-\epsilon)\gamma_{p}^{p} - p\delta (1+\epsilon) n^{(p-1)/2} (1+o(1))} > 3/4.
\]
Overall, choosing
\[
\epsilon \defeq \frac{\eta^{2}}{8}, \quad \delta \defeq \frac{\eta^{2} \gamma_{p}^{p}}{(8+\eta^{2}) p n^{(p-1)/2}}, \quad \text{and} \quad m > \max\left\{\frac{4 \delta^{-n} (\gamma_{2p}^{2p} - \gamma_{p}^{2p})}{\epsilon^{2} \gamma_{p}^{2p}}, \frac{5}{\epsilon^{2}}\right\},
\]
we get
\[
\prob{\min_{\norm{z}_{2}=1} \frac{1}{m} \sum_{i=1}^{m} \abs{\mydot{a_{i}}{z}}^{p} \geq \left(1 - \frac{\eta^{2}}{4}\right)\gamma_{p}^{p}} > \frac{1}{2}.
\]
On the other hand to analyze the value of the corresponding convex relaxation, we use
\begin{align*}
\expec{\norm{a_{i}}_{2}^{p}} & = \int_{a \in \R^{n}} \norm{a}_{2}^{p} \mu(a) da = n^{\nfrac{p}{2}} (1+o(1)) \\
\var{\norm{a_{i}}_{2}^{p}} & = \expec{\norm{a_{i}}_{2}^{2p}} - \expec{\norm{a_{i}}_{2}^{p}}^{2} = n^{p} (1+o(1))
\end{align*}
Again by Chebyshev's Inequality,
\[
\prob{\frac{1}{m} \sum_{i=1}^{m} \norm{a_{i}}_{2}^{p} \geq \left(1 + \nfrac{\eta}{2}\right) n^{\nfrac{p}{2}} (1+o(1))} \leq \frac{4(1+o(1))}{m\eta^{2}}.
\]
Choosing $m > \nfrac{9}{\eta^{2}}$, we get
\[
\prob{\frac{1}{m} \sum_{i=1}^{m} \norm{a_{i}}_{2}^{p} \leq \left(1 + \nfrac{\eta}{2}\right) n^{\nfrac{p}{2}} (1+o(1))} \geq 1 - \frac{4(1+o(1))}{m\eta^{2}} > \frac{1}{2}.
\]
Therefore, the convex relaxation satisfies
\[
\prob{\frac{1}{m} \sum_{i=1}^{m} \abs{\nfrac{1}{n} \cdot a_{i}^{T} I a_{i}}^{\nfrac{p}{2}} \leq \left(1 + \nfrac{\eta}{2}\right) (1+o(1))} > \frac{1}{2}.
\]
Hence,
\begin{align*}
& \prob{\min_{\norm{z}_{2}=1} \left(\frac{1}{m} \sum_{i=1}^{m} \abs{\mydot{a_{i}}{z}}^{p}\right)^{\nfrac{1}{p}} \geq (1-\eta)\cdot  \gamma_{p} \cdot \min_{\substack{I \succcurlyeq X \succcurlyeq 0 \\ \tr(X)=1}} \left(\frac{1}{m} \sum_{i=1}^{m} \abs{a_{i}^{T} X a_{i}}^{\nfrac{p}{2}}\right)^{\nfrac{1}{p}}} \\
& \geq \prob{\min_{\norm{z}_{2}=1} \frac{1}{m} \sum_{i=1}^{m} \abs{\mydot{a_{i}}{z}}^{p} \geq \left(1 - \nfrac{\eta^{2}}{4}\right)\cdot \gamma_{p}^{p} \quad \text{and} \quad \frac{1}{m} \sum_{i=1}^{m} \abs{\nfrac{1}{n} \cdot a_{i}^{T} I a_{i}}^{\nfrac{p}{2}} \leq \left(1 + \nfrac{\eta}{2}\right) (1+o(1))} \\
& > 0.
\end{align*}
\end{proof}

\section{NP-hardness of Subspace Approximation}\label{sec:np-hard}
\renewcommand{\thetheorem}{B.\arabic{theorem}}
\setcounter{theorem}{0}  % reset counter

In this section, we show unconditionally that the problem \subspace{$n-1$}{$p$} is NP-hard, for $p > 2$, using a reduction from the Min-Uncut problem on 
graphs. Such a result was also obtained independently by Gibson and Xiao 
(personal communication).

\noindent {\sc Min-Uncut problem:} Given a graph $G=(V, E)$, find a bipartition of its vertices $V = S \cup T$ that minimizes the number of edges with both endpoints on the same side of the bipartition.

Let $\size{V} = n$ and $\size{E} = m$. Min-Uncut problem is known to be NP-hard, i.e., for some $1 \leq t \leq m$ it is NP-hard to find if the Min-Uncut has at most $t$ edges. We give a polynomial time reduction from Min-Uncut to subspace approximation as follows: Given an instance of Min-Uncut, construct a matrix $A \in \R^{(m+n) \times n}$ such that
\[
\min_{\norm{y}_{2} = \sqrt{n}} \norm{Ay}_{p}^{p} = \min_{\norm{y}_{2} = \sqrt{n}} \sum_{ij \in E} (y_{i} + y_{j})^{p} + N \sum_{i=1}^{n} y_{i}^{p},
\]
where $N$ is an integer polynomially large in $n$ and $m$ which will be chosen later.

\noindent {\sf Yes case:} The Min-Uncut has at most $t$ edges. Define $x_{i} = 1$ if $i \in S$ and $x_{i} = 1$ if $i \in T$. Using this $x \in \{-1, 1\}^{n}$ we get $\mathrm{OPT} \leq \norm{Ax}_{p}^{p} = t2^{p} + Nn$.

\noindent {\sf No case:} Otherwise, for any bipartition the Min-Uncut has at least $t+1$ edges, i.e., for any $x \in \{-1, 1\}^{n}$ we have $\sum_{ij \in E} (x_{i} + x_{j})^{p} \geq (t+1)2^{p}$. Now divide the sphere of radius $\sqrt{n}$ into two parts as follows:
\begin{align*}
S & = \{y \suchthat \norm{y}_{2} = \sqrt{n}~ \text{and}~ \abs{y_{j}} \in (1-\epsilon, 1+\epsilon)~ \text{for all}~ j \in [n]\}, \\
T & = \{y \suchthat \norm{y}_{2} = \sqrt{n}~ \text{and}~ y \notin S\},
\end{align*}
where $\epsilon < \nfrac{1}{p} \cdot (m+1)$. For any $y \in T$,
\begin{itemize}
\item Case $1$: $\abs{y_{i}} = 1 + \epsilon_{i} \geq 1 + \epsilon$ for some $i$. Then,
\begin{align*}
\sum_{j=1}^{n} y_{j}^{p} & \geq (1+\epsilon_{i})^{p} + (n-1) \left(\frac{n - (1+\epsilon_{i})^{2}}{n-1}\right)^{\nfrac{p}{2}} \\
& \geq (1+\epsilon_{i})^{p} + (n-1) \left(1 - \frac{2\epsilon_{i} + \epsilon_{i}^{2}}{n-1}\right)^{\nfrac{p}{2}} \\
& \geq 1 + p\epsilon_{i} + {p \choose 2} \epsilon_{i}^{2} + (n-1) \left(1 - \nfrac{p}{2} \cdot  \frac{2 \epsilon_{i} + \epsilon_{i}^{2}}{n-1}\right) \\
& \geq n + \frac{p^{2} \epsilon^{2}}{4} \qquad \text{using $p > 2\left(1+\frac{1}{n-1}\right)$ for large enough $n$}.
\end{align*}
\item Case $2$: $\abs{y_{i}} = 1 - \epsilon_{i} \leq 1 - \epsilon$ for some $i$. Then,
\[
\sum_{j \neq i} y_{j}^{2} = n - (1-\epsilon_{i}^{2}).
\]
Hence, there exists some $k$ such that
\[
y_{k}^{2} \geq \frac{n - (1-\epsilon_{i})^{2}}{n-1} \geq 1 + \frac{2\epsilon_{i} - \epsilon_{i}^{2}}{n-1} \geq 1 + \frac{\epsilon}{n} \Rightarrow \abs{y_{k}} \geq 1 + \frac{\epsilon}{2n}.
\]
Therefore, using the same analysis as in the previous case, we get
\[
\sum_{j=1}^{n} y_{j}^{p} \geq n + \frac{p^{2} \epsilon^{2}}{16n^{2}}.
\]
\end{itemize}
Using the above property of $y \in T$, we get
\begin{align*}
\sum_{ij \in E} (y_{i} + y_{j})^{p} + N \sum_{j=1}^{n} y_{j}^{p} & \geq N \sum_{j=1}^{p} y_{j}^{p} \\
& \geq Nn + \frac{Np^{2}\epsilon^{2}}{16n^{2}} \\
& > t2^{p} + Nn \qquad \text{using $N > 2^{p+4}n^{2}m(m+1)^{2}$}.
\end{align*}
For any $y \in S$,
\begin{align*}
\sum_{ij \in E} (y_{i} + y_{j})^{p} + N \sum_{j=1}^{n} y_{j}^{p} & \geq (1-\epsilon)^{p}(t+1)2^{p} + Nn \\
& \geq (1-p\epsilon)(t+1) 2^{p} + Nn \\
& > t2^{p} + Nn \qquad \text{using $\epsilon < \frac{1}{p(t+1)}$}.
\end{align*}
%The above proof gives NP-hardness even for $p = \Theta(\log n)$. What does this imply for the $\ell_{\infty}$ case of %subspace approximation?
%
%{\sf Remark: Can we do better analysis to get hardness of approximation within a constant factor? One approach is to %use tensor product trick, the other is to use a different problem in the above reduction.}
%
\end{document}